\documentclass[12pt]{article}

\usepackage{a4}
\renewcommand{\title}[1]{
\begin{center} \Large \bf #1 \end{center}
}

\renewcommand{\author}[2]{
 \begin{center} #1  \vspace{3mm} \\
  #2 \\
 \end{center}
\addvspace{\baselineskip}
}

\usepackage{amssymb}
\usepackage{amsmath}
\usepackage{hyperref}

\usepackage{amsthm}
\newtheorem{thm}{Theorem}[section]
\newtheorem{prop}[thm]{Proposition}

\newtheorem{lem}[thm]{Lemma}
\newtheorem{conj}[thm]{Conjecture}

\theoremstyle{definition}
\newtheorem{defn}[thm]{Definition}

\theoremstyle{remark}

\makeatletter
\@addtoreset{equation}{section}

\makeatother

\allowdisplaybreaks[4]

\begin{document}

\baselineskip 5mm

\title{Exact solution of matricial $\Phi^3_2$ quantum field theory}

\author{Harald Grosse${}^1$, Akifumi Sako${}^{1,2}$ 
and Raimar Wulkenhaar${}^3$}{
${}^1$
Fakult\"at f\"ur Physik, Universit\"at Wien\\
Boltzmanngasse 5, A-1090 Wien, Austria\\[\smallskipamount]
${}^2$  Department of Mathematics,
Faculty of Science Division II,\\
Tokyo University of Science,
1-3 Kagurazaka, Shinjuku-ku, Tokyo 162-8601, Japan\\[\smallskipamount]
${}^3$ {Mathematisches Institut der Westf\"alischen
  Wilhelms-Universit\"at\\
Einsteinstra\ss{}e 62, D-48149 M\"unster, Germany}}

\noindent
{\bf MSC 2010:} 81T16, 81R12, 45F05
\vspace{1cm}

\footnotetext[1]{harald.grosse@univie.ac.at, 
$^2$sako@rs.tus.ac.jp, $^3$raimar@math.uni-muenster.de}

{\abstract \noindent
We apply a recently developed method to exactly solve 
the $\Phi^3$ matrix model with covariance of a 
two-dimensional theory, also known as regularised Kontsevich model. 
Its correlation functions collectively describe graphs on a multi-punctured 
2-sphere. We show how Ward-Takahashi identities and Schwinger-Dyson equations
lead in a special large-$\mathcal{N}$ limit to integral equations
that we solve exactly for all correlation functions. 
Remarkably, these functions are analytic in the $\Phi^3$ coupling constant,
although bounds on individual graphs justify only Borel
summability.
\\
\hspace*{1.5em}The solved model arises from noncommutative field theory 
in a special limit of strong deformation parameter.
The limit defines ordinary 2D Schwinger
functions which, however, do not satisfy reflection positivity.}

\section{Introduction}

Matrix models \cite{Di Francesco:1993nw} were intensely studied 
around 1990. Highlights include
the non-perturbative solution of the Hermitean one-matrix model
\cite{Brezin:1990rb, Douglas:1989ve, Gross:1989vs} and the
understanding that it gives a rigorous meaning to quantum gravity in
two dimensions. As proved by Kontsevich
\cite{Kontsevich:1992ti}, there is an equivalent formulation by a
model for Hermitean matrices $\Phi$ with action
$\mathrm{tr}(E\Phi^2+\frac{\mathrm{i}}{6} \Phi^3)$, where $E$ is a
fixed external matrix. Equivalently, the external structure can be
moved to the linear term. The resulting partition function
\begin{align}
\mathcal{Z}[J]=\int_{M_{\mathcal{N}}(\mathbb{C})} 
\mathcal{D}\Phi\;\exp\Big(-\mathrm{tr}\Big(-J\Phi
+ \frac{\mathcal{N}}{2} \beta \Phi^2
+ \frac{\mathcal{N}}{3} \alpha \Phi^3\Big)\Big)
\label{action-MS}
\end{align}
(all matrices self-adjoint) was solved by Makeenko 
and Semenoff \cite{Makeenko:1991ec}. The
strategy consists in a diagonalisation of $\Phi$ thanks to the 
Itzykson--Zuber--Harish-Chandra formula, leaving an integral over the
eigenvalues $x_i$ of the random matrix $\Phi$. Since these $x_i$ are
dummy integration variables, the partition function is invariant under
variations $x_i \mapsto x_i+\epsilon_n x_i^{n+1}$. These give rise to
Virasoro constraints on $\mathcal{Z}[j_1,\dots,j_N]$, which 
Makeenko-Semenoff were able to solve.

A renewed interest in matrix models came from field theories on
noncommutative spaces of Moyal-Weyl type. 
We mention the magnetic field model studied in \cite{Langmann:2003cg},
which is also exactly solvable but trivial as a field theory. 
The field theory of the $\Phi^3$ model on Moyal space with 
harmonic term (see below) 
has been studied by one of us (HG) and H.\ Steinacker 
in \cite{Grosse:2005ig, Grosse:2006tc}. The novel aspect was a renormalisation
procedure for the Kontsevich model. Only partial information
on correlation functions were obtained; this is the point where the
present paper goes much further.

Two of us (HG+RW) worked on the $\Phi^4$-theory on four-dimensional
Moyal-Weyl deformed space and cured the ultraviolet-infrared mixing by
adding a harmonic oscillator potential to the action. This leads to a
renormalisable model \cite{Grosse:2004yu}, which develops a zero of
the $\beta$-function of the coupling constant \cite{Disertori:2006nq}
at a special value of the parameter space.  At this special point the
model becomes a dynamical matrix model. In \cite{Grosse:2012uv} we
(HG+RW) extended the idea of \cite{Disertori:2006nq} to an alternative
solution strategy for matrix models, avoiding the diagonalisation
(which is useless for the $\Phi^4$ interaction). We used instead the
Ward-Takahashi identities which result from a variation $\Phi \mapsto
U^*\Phi U$, with $U=\exp(\mathrm{i}\epsilon B)$ unitary, to derive a
different type of Schwinger-Dyson equations. We proved that one of
them consists in a non-linear singular integral equation for the
2-point function alone (first obtained in \cite{Grosse:2009pa}), which
then determines all higher correlation functions. We subsequently
reduced the problem to a fixed point equation for a single function on
$\mathbb{R}_+$ and proved that a solution exists
\cite{Grosse:2015fka}. If one could prove that the solution is the
Stieltjes transform of a positive measure, which is true for the
computer \cite{Grosse:2014lxa}, then one could convert the model into
a 4-dimensional Euclidean quantum field theory with
reflection-positive Schwinger 2-point function \cite{Grosse:2013iva}.

\smallskip

In this paper we apply the strategy of \cite{Grosse:2012uv} to the
$\Phi^3_2$ matrix model\footnote{In our subsequent paper \cite{Grosse:2016qmk}
we extend this work to four and six dimensions. Whereas the renormalisation
of $\Phi^3_4$ and $\Phi^3_6$ is much more involved, the solution
of the Schwinger-Dyson equations is easily adapted from the 
$\Phi^3_2$ case. To avoid duplication of material we introduce 
in some formulae parameters $Z,\nu$ which 
at the end are set to $Z=1$ and $\nu=0$ for $\Phi^3_2$.}.
Since a linear term would be generated by
loop corrections, we add it from the beginning. 
We define first the model with cut-offs and give next
Ward-Takahashi (WT) identities and Schwinger-Dyson (SD)
equations. The 1-point function requires renormalisation,
after which the cut-off can be sent to $\infty$ in the usual way 
\cite{Makeenko:1991ec}; for noncommutative field theory 
this corresponds to a limit of large matrices coupled with an infinitely
strong deformation parameter --  a limit which is called the ``Swiss
cheese limit''. This way one projects onto the genus zero sector, but
keeps all possible boundary components. In this limit
the infinite hierarchy of SD-equations decouples (as in the 
$\Phi^4$-model \cite{Grosse:2012uv}). We find that a function $W(X)$ 
related to the 1-point function satisfies a non-linear integral 
equation which, up to the
renormalisation problem, \emph{is identical} to an equation 
solved by Makeenko-Semenoff \cite{Makeenko:1991ec}
in the framework of the Kontsevich model. This coincidence is by no means
surprising! We then proceed by resolving the entire hierarchy 
of linear equations for all genus-zero matrix correlation functions. 
Here combinatorial identities on Bell polynomials
play a crucial r\^ole. 

\smallskip

In the final section we relate the $\Phi^3$ matrix model to field
theory on noncommutative Moyal space. 
We also perform in position space the limit of large deformation
parameter. In this way a Euclidean quantum field theory on standard
(undeformed) $\mathbb{R}^2$ is obtained for which we can explicitly
describe all connected Schwinger functions. We deduce that already the
Schwinger 2-point function \emph{does not fulfil} reflection
positivity for whatever (real or imaginary) non-zero coupling
constant. This is in sharp contrast with the $\phi^4$-model where
numerical and partial analytic evidence was given that the
Schwinger 2-point function is reflection positive.

\smallskip

Associating a quantum field theory with a matrix model is somewhat
unusual in the traditional setup. We therefore begin in section 
\ref{prelude} with a description of this relation, thereby giving a 
precise definition of
correlation functions on the multi-punctured sphere, with 
$N_\beta$ fields attached to the $\beta^\text{th}$ boundary component
(=\,puncture). We also point out that from the graphical perspective the
perturbation series cannot be expected to converge; it is at best 
Borel summable. This 
highlights our achievement of explicit analytic formulae for any
correlation function.

\section{Prelude: A QFT toy model}
\label{prelude}

We consider planar graphs $\Gamma$ on the 2-sphere with two sorts of
vertices: any number of black (internal) vertices of valence 3, and
$B\geq 1$ white vertices $\{v_\beta\}_{\beta=1}^B$ (external vertices,
or punctures, or boundary components) of any valence $N_\beta\geq 1$.
Every face is required to have at most one white vertex (separation of
punctures). Faces with a white vertex are called external;
they are labelled by positive real numbers
$x^1_1,\dots,x^1_{N_1},\dots,x^B_1,\dots,x^B_{N_B}$ (the upper index
labels the unique white vertex of the face).  Faces without white
vertex are called internal; they are labelled by positive real numbers
$y_1,\dots, y_L$.  Such graphs are dual to triangulations of the
$B$-punctured sphere.

We associate
a weight $(-\tilde{\lambda})$ to each black vertex, weight $1$ to each 
white vertex, and weight $\frac{1}{z_1+z_2+1}$ to an edge separating
faces labelled by $z_1$ and $z_2$. These can be internal or external,
also $z_1=z_2$ can occur. Multiply the weights of all edges and vertices of the graph and integrate over all internal face variables $y_1,\dots,y_L$ 
from $0$ to a cut-off $\Lambda^2$, thus giving rise to a function 
$\tilde{G}^\Lambda_\Gamma(x^1_1,\dots,x^1_{N_1}|\dots|x^B_1,\dots,x^B_{N_B})$ of 
the external face variables. 

Three examples are in order:
\begin{align}
\Gamma_1: && \parbox{20mm}{\begin{picture}(20,10)
\put(15,5){\circle{10}}
\put(0.5,5){\line(1,0){9.5}}
\put(-1,4){\mbox{$\circ$}}
\put(9,4){\textbullet}
\put(3,6.5){\mbox{\small$x_1^1$}}
\put(13,5){\mbox{\small$y_1$}}
\end{picture}}
&&
\tilde{G}^\Lambda_{\Gamma_1}(x^1_1) &=\frac{(-\tilde{\lambda})}{2x^1_1+1}
\int_0^{\Lambda^2} \frac{dy_1}{x^1_1+y_1+1},
\\
\Gamma_2: && \parbox{24mm}{\begin{picture}(24,15)
\put(12,10){\circle{10}}
\qbezier(7,10)(0,10)(0,5)
\qbezier(17,10)(24,10)(24,5)
\qbezier(0,5)(0,0)(11.3,0)
\qbezier(24,5)(24,0)(12.6,0)
\put(6,9){\textbullet}
\put(16,9){\textbullet}
\put(11,-1){\mbox{$\circ$}}
\put(2.5,5){\mbox{\small$x_1^1$}}
\put(1.5,12){\mbox{\small$x_2^1$}}
\put(10,10){\mbox{\small$y_1$}}
\end{picture}}
&&
\tilde{G}^\Lambda_{\Gamma_2}(x^1_1,x^1_2) &=\frac{(-\tilde{\lambda})^2 }{
(x^1_1+x^1_2+1)^2}
\int_0^{\Lambda^2} \!\! \frac{dy_1}{(x^1_1+y_1+1)(x^1_2+y_1+1)},
\\
\Gamma_3: && \parbox{24mm}{\begin{picture}(24,10)
\put(16,5){\oval(16,10)}
\put(0.5,5){\line(1,0){8}}
\put(24,5){\line(-1,0){6}}
\put(7,4){\textbullet}
\put(23,4){\textbullet}
\put(-1,4){\mbox{$\circ$}}
\put(16.5,4){\mbox{$\circ$}}
\put(2,6.5){\mbox{\small$x_1^1$}}
\put(11.5,5){\mbox{\small$x^2_1$}}
\end{picture}}
&&
\tilde{G}^\Lambda_{\Gamma_3}(x^1_1|x^2_1) &=\frac{(-\tilde{\lambda})^2 }{
(2x^1_1+1)(2x^2_1+1)(x^1_1+x^2_1+1)^2}.
\end{align}

This setting defines a toy model of quantum field theory, sharing all
typical features. It has the power-counting behaviour of the
$\Phi^3_2$ model, in particular has a single divergence: The limit
$\lim_{\Lambda\to \infty}\tilde{G}^\Lambda_{\Gamma_1}(x^1_1)$ does not
exist. The problem is cured by renormalisation. We assume the reader
is familiar with the notion of one-particle irreducible (1PI)
subgraphs.  The renormalisation of the toy quantum field theory
consists in recursively replacing all 1PI one-point subfunctions
$f(z)$ by its Taylor subtraction $f(z)-f(0)$. This does more than
necessary, but permits the global (i.e.\ non-perturbative)
normalisation rule $\tilde{G}_{\Gamma}(0)=0$ for any graph $\Gamma$
with a single white vertex of valence $1$.  Omitting the superscript
$\Lambda$ on $\tilde{G}$ means recursive renormalisation plus 
limit $\Lambda\to \infty$.  We note
\begin{align}
\tilde{G}_{\Gamma_1}(x^1_1)  =\frac{(-\tilde{\lambda})}{2x^1_1+1}
\int_0^\infty dy_1 \Big(\frac{1}{x^1_1+y_1+1}-\frac{1}{y_1+1}\Big)
= \tilde{\lambda} \frac{\log(x^1_1+1)}{2x^1_1+1} .
\end{align}

Consider the following challenge: Fix $B$ white vertices
of valences $N_1,\dots,N_B$, take an arbitrary number (there is a
lower bound) of black vertices, and connect them in all possible ways
to planar graphs. Assign the weights, perform the renormalisation,
evaluate the face integrals (for $\Lambda\to \infty$) and sum
everything up. What does this give?  

One meets here a main difficulty 
of quantum field theory: there are too many graphs.
The number of connected planar graphs with $n$ black vertices can be
estimated by the number $n^{n-2}$ of ordered trees with $n$
vertices. With the typical tools of quantum field theory, see e.g.\ 
\cite{Rivasseau:1991ub}, one can prove 
uniform bounds of the type $|\tilde{G}_\Gamma|
\leq C_1 \cdot |\tilde{\lambda}|^n C_2^n$. This allows to give a
meaning to $
\tilde{G}(x^1_1,\dots,x^1_{N_1}|\dots|x^B_1,\dots,x^B_{N_B})
=\sum_\Gamma \tilde{G}_\Gamma
(x^1_1,\dots,x^1_{N_1}|\dots|x^B_1,\dots,x^B_{N_B})$ as a Borel
resummation, where $\lambda$ belongs to a sufficiently small disk
tangent to the imaginary axis. Absolute convergence is impossible for
any $\tilde{\lambda}\neq 0$. We should remark that more
complicated QFT models have an additional renormalon problem which
excludes even Borel summability. In such case one has to employ 
the constructive renormalisation machinery 
 \cite{Rivasseau:1991ub} with its infinitely many (but mutually 
 related) effective coupling constants.

We hope that the reader, with these remarks in mind, will appreciate that 
we will provide \emph{exact formulae} for any
$
\tilde{G}(x^1_1,\dots,x^1_{N_1}|\dots|x^B_1,\dots,x^B_{N_B})$. 
Remarkably, these functions are \emph{analytic} in $\tilde{\lambda}^2$!
For convenience we refer to the simplest cases:
$\tilde{G}(x^1_1)$ will be given in (\ref{Gx-final}),
$\tilde{G}(x^1_1,x^1_2)$ implicitly in (\ref{Gxy}) and 
$\tilde{G}(x^1_1|x^2_1)$ implicitly in (\ref{G1+1}). One has to insert 
$X^\beta_i=(2x^\beta_i+1)^2$ and the formulae for $W(X)$ and 
$c(\tilde{\lambda})$ given in Proposition \ref{prop_MS}.
The order-$n$ Taylor term reproduces the sum of all graphs 
with $n$ black vertices and $B$ white vertices of valences $N_1,\dots,N_B$.
The reader is invited to convince herself/himself that these 
formulae (restricted to the relevant order in $\tilde{\lambda}$) 
and the graphical rules agree on the the following examples:
\begin{align}
&\tilde{G}_{(3)}(x^1_1) = 
\parbox{22mm}{~\begin{picture}(22,8)
\put(0.5,4){\line(1,0){4.5}}
\put(8.5,4){\circle{7}}
\put(11.5,4){\line(1,0){3}}
\put(18.5,4){\circle{7}}
\put(7,4){\mbox{\small$y_2$}}
\put(17,4){\mbox{\small$y_1$}}
\put(0,6){\mbox{\small$x^1_1$}}
\put(4,3){\textbullet}
\put(11,3){\textbullet}
\put(14,3){\textbullet}
\put(-1,3){\mbox{$\circ$}}
\end{picture}}
{}~+~
\parbox{25mm}{\begin{picture}(25,16)
\put(0.5,8){\line(1,0){5.5}}
\put(15,8){\oval(18,16)}
\put(16,8){\circle{8}}
\put(24,8){\line(-1,0){4}}
\put(14,8){\mbox{\small$y_1$}}
\put(9,12){\mbox{\small$y_2$}}
\put(1,10){\mbox{\small$x^1_1$}}
\put(5,7){\textbullet}
\put(23,7){\textbullet}
\put(19.4,7){\textbullet}
\put(-1,7){\mbox{$\circ$}}
\end{picture}}
{}~+~
\parbox{20mm}{\begin{picture}(20,16)
\put(0.5,8){\line(1,0){7.5}}
\put(8,8){\line(2,1){4}}
\put(8,8){\line(2,-1){4}}
\put(14.5,12.5){\circle{7}}
\put(14.5,3.5){\circle{7}}
\put(13,3){\mbox{\small$y_1$}}
\put(13,12){\mbox{\small$y_2$}}
\put(1,10){\mbox{\small$x^1_1$}}
\put(7,7){\textbullet}
\put(11.2,9){\textbullet}
\put(11.2,5){\textbullet}
\put(-1,7){\mbox{$\circ$}}
\end{picture}}
{}~+~
\parbox{25mm}{\begin{picture}(25,16)
\put(0.5,8){\line(1,0){5.5}}
\put(15,8){\oval(18,16)}
\put(15,0){\line(0,1){16}}
\put(18,8){\mbox{\small$y_1$}}
\put(9,8){\mbox{\small$y_2$}}
\put(1,10){\mbox{\small$x^1_1$}}
\put(5,7){\textbullet}
\put(14,-1){\textbullet}
\put(14,15){\textbullet}
\put(-1,7){\mbox{$\circ$}}
\end{picture}}
\nonumber
\\
&=\tilde{\lambda}^3\Big(
\frac{(\log 2)^2}{2x^1_1+1}
-\frac{(\log 2)^2}{(2x^1_1+1)^3}\Big),
\label{Feyn-G1}
\\[2ex]
&\tilde{G}_{(2)}(x^1_1,x^1_2) = 
\parbox{24mm}{\begin{picture}(24,15)
\put(12,10){\circle{10}}
\qbezier(7,10)(0,10)(0,5)
\qbezier(17,10)(24,10)(24,5)
\qbezier(0,5)(0,0)(11.3,0)
\qbezier(24,5)(24,0)(12.6,0)
\put(6,9){\textbullet}
\put(16,9){\textbullet}
\put(11,-1){\mbox{$\circ$}}
\put(2.5,5){\mbox{\small$x_1^1$}}
\put(1.5,12){\mbox{\small$x_2^1$}}
\put(10,10){\mbox{\small$y_1$}}
\end{picture}}
{}~+~ 
\parbox{24mm}{\begin{picture}(24,15)
\put(12,12){\circle{6}}
\qbezier(12,6)(0,6)(0,3)
\qbezier(12,6)(24,6)(24,3)
\qbezier(0,3)(0,0)(11.3,0)
\qbezier(24,3)(24,0)(12.6,0)
\put(12,6){\line(0,1){3}}
\put(11,5){\textbullet}
\put(11,8){\textbullet}
\put(11,-1){\mbox{$\circ$}}
\put(4,2){\mbox{\small$x_1^1$}}
\put(4,8){\mbox{\small$x_2^1$}}
\put(10,11.5){\mbox{\small$y_1$}}
\end{picture}}
{}~+\quad 
\parbox{24mm}{\begin{picture}(24,15)
\put(12,9){\circle{6}}
\qbezier(12,15)(0,15)(0,7.5)
\qbezier(12,15)(24,15)(24,7.5)
\qbezier(0,7.5)(0,0)(11.3,0)
\qbezier(24,7.5)(24,0)(12.6,0)
\put(12,15){\line(0,-1){3}}
\put(11,14){\textbullet}
\put(11,11){\textbullet}
\put(11,-1){\mbox{$\circ$}}
\put(-3,1){\mbox{\small$x_2^1$}}
\put(4,4){\mbox{\small$x_1^1$}}
\put(10,8.5){\mbox{\small$y_1$}}
\end{picture}}
\nonumber
\\[1ex]
&= 
\frac{\tilde{\lambda}^2}{(x^1_1+x^1_2+1)^2} 
\Big(
\frac{\log(x^1_1+1)-\log(x^1_2+1)}{x^1_1-x^1_2}
-\frac{\log(x^1_1+1)}{2x^1_1+1}
-\frac{\log(x^1_2+1)}{2x^1_2+1}\Big).
\end{align}

In fact we solve a more general case with weight functions
$\frac{1}{e(z_1)+e(z_2)+1}$ for the edges, where $e:\mathbb{R}_+\to
\mathbb{R}_+$ is a differentiable function of positive
derivative. Equivalently, one can keep the old face variables $y_i$
but assign a weight $\tilde{\rho}(y_i)= \frac{1}{e'(e^{-1}(y_i))}$ to
the faces. The asymptotic behaviour of $\tilde{\rho}(y)\sim
y^{\frac{D}{2}-1}$ for $y\to \infty$ encodes a dimensionality $D$,
where actually only the even integer $2[\frac{D}{2}]$ matters. This
paper treats $2[\frac{D}{2}]=2$. For $2[\frac{D}{2}]=0$ we have a
finite model where no renormalisation is necessary. In
\cite{Grosse:2016qmk} we extend this work to $2[\frac{D}{2}]=4$ (which
also has a finite number of divergences) and to the just
renormalisable case $2[\frac{D}{2}]=6$.


\section{The setup}

\label{sect2}

Consider the following action functional for 
Hermitean matrix-valued `fields' 
$\Phi =\Phi^*\in M_{\mathcal{N}}( {\mathbb C})$:
\begin{align}
\label{action-MM}
S&= V\, {\rm tr} ( E\Phi^2 +\kappa \Phi 
+ \frac{\lambda}{3} \Phi^3) ,
\end{align}
or explicitly (in symmetrised form)
\begin{align}
S&= V \Big(\sum_{n,m=0}^{\mathcal{N}}
\frac{1}{2} \Phi_{nm} \Phi_{mn} H_{nm} +\kappa \sum_{m=0}^{\mathcal{N}}  \Phi_{mm} +
\frac{\lambda}{3} \sum_{k,l,m=0}^{\mathcal{N}}  \Phi_{kl} \Phi_{lm} \Phi_{mk} \Big),
\nonumber
\\
H_{mn}&:= E_m+E_n.
\label{Hmn}
\end{align}
Here $V$ is a constant discussed later, 
$\lambda$ is the coupling
constant (real or complex), and $\kappa$ will be needed for renormalising
the 1-point function. The self-adjoint positive matrix 
$E=(E_m \delta_{mn})$ plays a crucial r\^ole.
We assume that the eigenvalues $E_m$ are a discretisation of 
a monotonously increasing differentiable function $e$ with $e(0)=0$,
\begin{align}
E_m=\mu^2\Big(\frac{1}{2}+e\Big(\frac{m}{\mu^2 V}\Big)\Big),
\label{E}
\end{align}
thus identifying $2E_0=\mu^2$ with a squared mass. The resulting 
covariance functions $\frac{1}{H_{mn}}=
\frac{1}{\mu^2( e(\frac{m}{\mu^2 V})+
e(\frac{n}{\mu^2 V})+1)}$ are nothing else than the (discretised) edge 
weights considered in section~\ref{prelude}. In particular, the 
discussion on the dimensionality encoded in $e$ 
(i.e.\ in the spectrum of $E$) applies. 

Comparison with (\ref{action-MS}) suggests that $V$ is proportional to
the size $\mathcal{N}$ of the matrices. This is precisely what we will
do. The only reason to keep them distinct is the fact that, as recalled 
in section~\ref{sec:NCG}, the action 
(\ref{action-MM}) naturally arises in
noncommutative field theory. There, $V$ is related to 
the deformation parameter, so that the limit $\mathcal{N}\sim V\to \infty$ 
defines the strong-deformation regime.

The partition function with an external field $J$, which is
also a self-adjoint matrix, is formally defined by
\begin{align}
\mathcal{Z}[J] &:= \int {\cal D}\Phi\; \exp \big( -S
+V\, {\rm tr} (J \Phi) \big)
\label{z_1}
\\
&= K \exp \Big( -\frac{\lambda}{3V^2} \sum_{m,n,k=0}^{\mathcal{N}} 
\frac{\partial^3}{\partial J_{mn}\partial J_{nk}\partial J_{km}}
\Big) \Big)
\mathcal{Z}_{free}[J],\nonumber
\\
\mathcal{Z}_{free}[J] &:= \exp \Big( \sum_{m,n=0}^{\mathcal{N}} \frac{V}{2}
(J_{nm}-\kappa \delta_{nm}) H^{-1}_{nm}
(J_{mn}-\kappa \delta_{nm}) \Big),
\label{z_free}
\end{align}
where $
K= \int {\cal D}\Phi \;\exp \big(-\frac{V}{2}  \sum_{m,n=0}^{\mathcal{N}}
\Phi_{mn} H_{mn} \Phi_{nm} \big)$.

A perturbative expansion of $\log \mathcal{Z}[J]$ gives 
exactly the graphical setup described in section~\ref{prelude} --
up to discretisation and temporary admission of non-planar graphs. 
The matrix indices correspond to face variables, edges between faces $m,n$ 
have weight $\frac{1}{H_{mn}}$, and the $\Phi^3$ vertices are the black 
ones with weight $(-\lambda)$. Identifying the white vertices is a 
little tricky. It turns out that the source matrices $J$ partition 
into cycles $\mathbb{J}_{p_1\dots
  p_{N_\beta}}:=\prod_{j=1}^{N_\beta} J_{p_jp_{j+1}}$, with $N_\beta+1\equiv 1$. 
Such a cycle of length $N_\beta$ is what we call a white vertex of 
valence $N_\beta$. Indeed, a `star' of covariances 
$\prod_{j=1}^{N_\beta} \frac{1}{H_{p_jp_{j+1}}}$ attaches to the 
source matrices, which graphically means that the white vertex is 
the common corner of the $N_\beta$ external faces labelled by 
$p_1,\dots,p_{N_\beta}$. 

With this identification we can 
represent $\log\mathcal{Z}$ as a sum over the number 
and the valences of the white vertices, i.e.\ the cycles of 
source matrices:
\begin{align}
\log\frac{ \mathcal{Z}[J]}{\mathcal{Z}[0]}
=:\sum_{B=1}^\infty \sum_{1\leq N_1 \leq \dots \leq
  N_B}^\infty
\sum_{p_1^1,\dots,p^B_{N_B} =0}^{\mathcal{N}} \!\!\!\!
V^{2-B}
&\frac{G_{|p_1^1\dots p_{N_1}^1|\dots|p_1^B\dots p^B_{N_B}|}
}{S_{(N_1,\dots ,N_B)}}
\prod_{\beta=1}^B \frac{\mathbb{J}_{p_1^\beta\dots
    p^\beta_{N_\beta}}}{N_\beta},
\label{logZ}
\end{align}
where the symmetry factor $S_{(N_1,\dots ,N_B)}$ is chosen as follows:
If we regroup identical valence numbers $N_\beta$ as
$(N_1,\dots,N_B)=(\underbrace{N'_1,\dots,N'_1}_{\nu_1},\dots,
\underbrace{N'_s,\dots,N'_s}_{\nu_s})$, then
$S_{(N_1,\dots ,N_B)}=\prod_{i=1}^{s} \nu_i!$.
The expansion coefficients
$G_{|p_1^1\dots p_{N_1}^1|\dots|p_1^B\dots p^B_{N_B}|}$ are called
$(N_1{+}\dots{+}N_B)$-point function. In principle they further expand into 
graphs $\Gamma$ with all possible numbers of black vertices and 
their connections. As pointed out in section~\ref{prelude}, the 
resummation is a problematic issue. We therefore keep the 
$(N_1{+}\dots{+}N_B)$-point functions intact and never expand into graphs. 
We will prove in this paper
(similarly to \cite{Grosse:2012uv}) that these functions have a well-defined
large-$(\mathcal{N},V)$ limit precisely for the given 
a scaling factor $V^{2-B}$
in $\log \mathcal{Z}[J]$.
For later purpose we note the first terms of the resulting expansion
of the partition function itself:
\begin{align}
\frac{ \mathcal{Z}[J]}{\mathcal{Z}[0]}
&= 1+ V \sum_m G_{|m|} J_{mm}
\label{explogZ}
\\
&+ \frac{V}{2} \sum_{m,n}
G_{|mn|} J_{mn}J_{nm} + \sum_{m,n}
\Big(\frac{1}{2}G_{|m|n|} + \frac{V^2}{2} G_{|m|}G_{|n|}\Big) J_{mm}J_{nn}
\nonumber
\\
&+
\frac{V}{3} \sum_{m,n,k}
G_{|mnk|} J_{mn}J_{nk}J_{km}
+\sum_{m,n,k}
\Big(\frac{1}{2} G_{|mn|k|}
+\frac{V^2}{2} G_{|mn|} G_{|k|}\Big)J_{mn}J_{nm}J_{kk}
\nonumber
\\
&\qquad + \sum_{m,n,k}
\Big(\frac{1}{6V}G_{|m|n|k|} +
\frac{V}{2} G_{|m|n|} G_{|k|} +
 \frac{V^3}{6} G_{|m|}G_{|n|}G_{|k|}\Big) J_{mm}J_{nn} J_{kk}
+
\dots.
\nonumber
\end{align}
All sums run from $0$ to a cut-off $\mathcal{N}$.

We repeat the remark pointed out in \cite{Grosse:2012uv} that these correlation
functions have common source factors on the diagonal, e.g.\
$(V^1 G_{|aa|}+G_{|a|a|})J_{aa}J_{aa}$. The functions
$G_{|aa|}$ and $G_{|a|a|}$ are clearly distinguished by their topology 
(number and valence of white vertices) 
and most conveniently identified by continuation of
$G_{|ab|}$ and $G_{|a|b|}$ to the diagonal.

Finally, we introduce our main tool: the Ward-Takahashi identities.
As proved in \cite{Disertori:2006nq, Grosse:2012uv}, the invariance of
the partition function (\ref{z_1}) under inner
automorphisms $\Phi\mapsto U^*\Phi U$ boils down to the WT-identities
\begin{align} 
\label{WT}
\sum_m \frac{\partial}{\partial J_{am} } \frac{\partial}{\partial J_{mb} } 
\mathcal{Z}[J] = \mathrm{W}_a \delta_{ab}+
\sum_m \frac{V}{E_a-E_b} \left(J_{ma}\frac{\partial}{\partial J_{mb} }
-  J_{bm}\frac{\partial}{\partial J_{am} } \right) \mathcal{Z}[J],
\end{align}
where the precise form of $\mathrm{W}_a$ (which we shall not need)
is given in \cite[Thm 2.3]{Grosse:2012uv}. These identities are 
exactly the counterpart of the Virasoro constraints in the traditional 
approach to matrix models \cite{Makeenko:1991ec}.

\section{Schwinger-Dyson equations and their 
solution for $B=1$}

\subsection{1- and 2-point functions}

We now derive a formula for the connected 1-point function $G_{|a|}$ by
inserting (\ref{z_1}), (\ref{z_free}) into the corresponding term of
(\ref{logZ}):
\begin{align}
G_{|a|}&=  \frac{\partial \log \mathcal{Z}[J]}{V\,\partial J_{aa}}
\Big|_{J=0}
= \frac{K}{\mathcal{Z}[0]}
\exp \Big( {-}\frac{\lambda}{3V^2} \!\sum_{m,n,k}\!
\frac{\partial^3}{\partial J_{mn}\partial J_{nk}\partial J_{km}}
\Big) 
\Big( (J_{aa}{-}\kappa) H_{aa}^{-1}\mathcal{Z}_{free}[J]\Big)\Big|_{J=0}
\nonumber
\\
&=H_{aa}^{-1}
\Big(-\kappa-\frac{\lambda}{V^2\mathcal{Z}[0]} \sum_{m=0}^{\mathcal{N}}
\frac{\partial}{\partial J_{am}}
\frac{\partial}{\partial J_{ma}} \mathcal{Z}[J]
\Big)\Big|_{J=0}
\nonumber
\\
&=H_{aa}^{-1}
\Big(-\kappa-\lambda G_{|a|}^2-\frac{\lambda}{V} \sum_{m=0}^{\mathcal{N}}
G_{|am|}-\frac{\lambda}{V^2} G_{|a|a|}
\Big)\;.
\label{1point_ward}
\end{align}
The last line follows from a two-fold differentiation of
(\ref{explogZ}). Of course the sum $\sum_{m=0}^{\mathcal{N}} G_{|am|}$ includes
$m=a$!

The connected 2-point function $G_{|ab|}$ is computed for $a\neq b$ as follows:
\begin{align}
G_{|ab|}&= \frac{\partial^2 \log \mathcal{Z}[J]}{V\,\partial J_{ab}
\partial J_{ba}}
\Big|_{J=0}
= \frac{K}{\mathcal{Z}[0]}
\exp \Big( {-}\frac{\lambda}{3V^2} \!\sum_{m,n,k}\!
\frac{\partial^3}{\partial J_{mn}\partial J_{nk}\partial J_{km}}
\Big) 
\frac{\partial}{\partial J_{ab}}
\Big( J_{ab} H_{ab}^{-1}\mathcal{Z}_{free}\Big)\Big|_{J=0}
\nonumber
\\
&=H_{ab}^{-1}-\frac{\lambda}{V^2}
\frac{H_{ab}^{-1}}{\mathcal{Z}[0]}\sum_{m=0}^{\mathcal{N}}
\frac{\partial}{\partial J_{ab}}
\frac{\partial}{\partial J_{bm}}
\frac{\partial}{\partial J_{ma}}
\mathcal{Z}[J]
\Big|_{J=0}
\nonumber
\\
&=H_{ab}^{-1}-\frac{\lambda}{V(E_b-E_a)}
\frac{H_{ab}^{-1}}{\mathcal{Z}[0]}\sum_{m=0}^{\mathcal{N}}
\frac{\partial}{\partial J_{ab}}
\Big( J_{mb}\frac{\partial \mathcal{Z}}{\partial J_{ma}}
-J_{am}\frac{\partial \mathcal{Z}}{\partial J_{bm}}
\Big)
\Big|_{J=0}
\nonumber
\\
&=H_{ab}^{-1}-\frac{\lambda}{V(E_b-E_a)}
\frac{H_{ab}^{-1}}{\mathcal{Z}[0]}
\Big( \frac{\partial \mathcal{Z}}{\partial J_{aa}}
-\frac{\partial \mathcal{Z}}{\partial J_{bb}}
\Big)
\Big|_{J=0}
\nonumber
\\
&=H_{ab}^{-1}\Big(1+\lambda
\frac{(G_{|a|}-G_{|b|})}{E_a-E_b}\Big)\;.
\label{2pointWard}
\end{align}
In the step from the 2nd to 3rd line we have used the Ward-Takahashi
identity (\ref{WT}).  The equation extends by continuity to $a=b$,
i.e.\ $G_{|aa|}=H_{aa}^{-1}+\lambda H_{aa}^{-1} \lim_{b\to
  a}\frac{(G_{|a|}-G_{|b|})}{E_a-E_b}$. The limit is well-defined in
perturbation theory where $G_{|a|}$ is, before performing the loop
sum, a rational function of the $E_n$ so that a factor $E_a-E_b$ can
be taken out of $G_{|a|}-G_{|b|}$. We shall later see that our
large-($\mathcal{N},V$) limit automatically gives a meaning also to
$\lim_{b\to a}$.

The na\"{\i}ve limit $\mathcal{N}\to \infty$ in (\ref{1point_ward})
will diverge unless $\kappa=\kappa(\mathcal{N})$ is carefully adjusted.
We chose a renormalisation condition
\begin{align}\label{A_cond}
G_0= 0 \qquad \Leftrightarrow\qquad
-\kappa(\mathcal{N})=  \frac{\lambda}{V} \sum_{m=0}^{\mathcal{N}} G_{0m}
+\frac{\lambda}{V^2} G_{|0|0|},
\end{align}
where a well-defined limit $G_{|00|}$ is assumed.
Substituting (\ref{2pointWard}) and (\ref{A_cond}) into (\ref{1point_ward}),
the Schwinger-Dyson equations are obtained as
\begin{align}
G_{|a|}
&=H_{aa}^{-1}
\Big\{-\lambda G_{|a|}^2-\frac{\lambda}{V} \sum_{m= 0}^{\mathcal{N}}
(H_{am}^{-1}-H_{0m}^{-1})
-\frac{\lambda}{V^2} (G_{|a|a|}-
G_{|0|0|})
\nonumber
\\
& -\frac{\lambda^2}{V}\sum_{m=0}^{\mathcal{N}}
\Big(
H_{am}^{-1}\,
\frac{(G_{|a|}-G_{|m|})}{E_a-E_m}
-H_{0m}^{-1}\,\frac{G_{|m|}}{E_m-E_0}
\Big)
\Big\}\;.\label{SD}
\end{align}
This equation suggests to introduce
\begin{align}
\frac{W_{|a|}}{2\lambda}:=G_{|a|}+\frac{H_{aa}}{2\lambda}
=G_{|a|}+\frac{E_{a}}{\lambda}.
\label{def:Wa}
\end{align}
Taking $H_{am}(E_a-E_m)=E_a^2-E_m^2$ into account, we arrive at
\begin{align}
W_{|a|}^2 &= 4 E_{a}^2
-
\frac{4\lambda^2}{V^2} (G_{|a|a|}-G_{|0|0|})
-\frac{2\lambda^2}{V}\sum_{m=0}^{\mathcal{N}}
\Big(
\frac{(W_{|a|}-W_{|m|})}{E_a^2-E_m^2}
-\frac{W_{|m|}-W_{|0|}}{E_m^2-E_0^2}
\Big),
\label{SD-W}
\\
G_{|ab|}&=\frac{1}{2} \frac{W_{|a|}-W_{|b|}}{E_a^2-E_b^2}\;.
\label{SD-G2}
\end{align}

\subsection{Large-($\mathcal{N},V$) limit
and integral equations}

Let us take the limit $\mathcal{N},V \rightarrow \infty$ subject to
fixed ratio
$\frac{\mathcal{N}}{V}=\mu^2 \Lambda^2$, in which the sum converges
to a Riemann integral
\begin{align}
\lim \frac{1}{V} \sum_{m=0}^{\mathcal{N}} f(m/V)
= \mu^2\Lambda^2 \int_0^1 du\;
f\big(\mu^2 \Lambda^2 u\big)
= \mu^2 \int_0^{\Lambda^2} dx\;f(\mu^2 x).
\end{align}
Expressing discrete matrix elements as $a=:V\mu^2 x$, the eigenvalues of 
$E$ take the form $E_a=\mu^2(e(x)+\frac{1}{2})$, see
(\ref{E}).  
We introduce the dimensionless\footnote{From 
the partition function (\ref{z_1}) and its expansion (\ref{logZ}) 
one reads off the following mass dimensions:
\begin{align*}
[\Phi]&=\mu^0, & [J]&=\mu^2,& [\kappa]&=\mu^2,&
[\lambda]&=\mu^2, &
[G_{|p_1^1\dots p_{N_1}^1|\dots|p_1^B\dots p^B_{N_B}|}]&=\mu^{2(2-B-N)}.
\end{align*}
\label{fn-1}} coupling
constant
$\tilde{\lambda}:=\frac{\lambda}{\mu^2}$ and define
\begin{align}
\mu^2 \tilde{W}(x):=\lim_{\mathcal{N},V\to \infty} W_{|V\mu^2x|},\qquad
\tilde{G}(x):=\lim_{\mathcal{N},V\to \infty} G_{|V\mu^2x|},
\end{align}
related by
$\frac{\tilde{W}(x)}{2\tilde{\lambda}}
=\tilde{G}(x)+\frac{e(x)+\frac{1}{2}}{
\tilde{\lambda}}$.
Now the limit of (\ref{SD-W}) becomes 
\begin{align}
(\tilde{W}(x))^2 &= (2e(x)+1)^2 
\label{Int_eq_field0}
\\
&-8\tilde{\lambda}^2 \int_0^{\Lambda^2} dy
\Big(
\frac{\tilde{W}(x)-\tilde{W}(y)}{(2e(x)+1)^2-(2e(y)+1)^2}-
\frac{\tilde{W}(y)-\tilde{W}(0)}{(2e(y)+1)^2-1}
\Big). \nonumber
\end{align}
We assume here $G_{|V\mu^2x|V\mu^2x|}=\mathcal{O}(V^0)$ so that this
term does not contribute to the limit; this will be checked later. It can be
seen graphically that this term generates higher genus
contributions, which are scaled away in the large-$\mathcal{N}$ limit.
A final transformation
\begin{align}
X:= (2e(x)+1)^2 ,\quad
W(X)=\tilde{W}(x(X)),\quad
G(X)=\tilde{G}(x(X)),\quad
\end{align}
and similarly for other capital letters $Y(y),T(y)$ and functions
$G(X,Y)=\tilde{G}(x(X),y(Y))$ etc.,
simplifies (\ref{Int_eq_field0}) to
\begin{align}
W^2(X)
+\int_1^{\Xi} dY \rho(Y) \,
\frac{W(X)-W(Y)}{X-Y}
&= X
+\int_1^{\Xi} dY \rho(Y) \,
\frac{W(1)-W(Y)}{1-Y},
\label{Int_eq_field1}
\\
\rho(Y):=\frac{2 \tilde{\lambda}^2}{\sqrt{Y}
\cdot e'(e^{-1}(\frac{\sqrt{Y}-1}{2}))}
,\qquad \Xi&:=(1+2e(\Lambda^2))^2.
\nonumber
\end{align}
Equation (\ref{Int_eq_field1}) closely resembles a problem
solved in the appendix of Makeenko-Semenoff \cite{Makeenko:1991ec}.
We take their solution (obtained by solving a Riemann-Hilbert problem) as
an ansatz\footnote{Our ansatz is more 
general than necessary in 2 dimensions. We need with $Z,\nu$ in 4 and
6 dimensions \cite{Grosse:2016qmk} and treat already here the 
general case in order to avoid
duplication in \cite{Grosse:2016qmk}.}
\begin{align}
W(X) := \frac{\sqrt{X+c}}{\sqrt{Z}}-\nu
+ \frac{1}{2} \int_1^\Xi dT \frac{\rho(T)}{(\sqrt{X+c} + \sqrt{T+c})\sqrt{T+c}}
\label{MS-solution}
\end{align}
with constants $Z,\nu,c$ 
determined by normalisation 
and consistency conditions
(thus becoming functions of $\lambda,\Xi$).
Straightforward
computation using $\frac{\sqrt{X+c}-\sqrt{Y+c}}{X-Y}=\frac{1}{
\sqrt{X+c}+\sqrt{Y+c}}$ yields
\begin{align}
&\int_1^{\Xi} dY \rho(Y) \,
\frac{W(X)-W(Y)}{X-Y}
\nonumber
\\*
&=\frac{\sqrt{X+c}}{\sqrt{Z}} 
\int_1^{\Xi} \frac{dY \rho(Y)}{(\sqrt{X+c}+\sqrt{Y+c})
\sqrt{X+c}}
\nonumber
\\*
& - \frac{1}{2}\int_1^{\Xi} \!\! \frac{dT \rho(T)}{\sqrt{T+c}
(\sqrt{X+c} + \sqrt{T+c})}
\int_1^{\Xi} \!\! \frac{dY \rho(Y)}{
(\sqrt{X+c} + \sqrt{Y+c})
(\sqrt{Y+c} + \sqrt{T+c})}.
\nonumber
\end{align}
In the last line we can symmetrise
$\frac{1}{\sqrt{T+c}}\mapsto \frac{1}{2}\big(
\frac{1}{\sqrt{T+c}}+\frac{1}{\sqrt{Y+c}}\big)$ so that the double
integral factors. Converting the second line by rational fraction
expansion, we arrive at
\begin{align}
\int_1^{\Xi} dY \rho(Y) \,
\frac{W(X)-W(Y)}{X-Y}
&=\frac{1}{\sqrt{Z}} \int_1^{\Xi} \frac{dY \rho(Y)}{\sqrt{Y+c}}
-\frac{1}{\sqrt{Z}} 
\int_1^{\Xi} \frac{dY \rho(Y)\;\sqrt{X+c}}{\sqrt{Y+c}
(\sqrt{X+c}+\sqrt{Y+c})}
\nonumber
\\
& - \frac{1}{4}
\Big(\int_1^{\Xi} \frac{dT \rho(T)}{\sqrt{T+c}
(\sqrt{X+c} + \sqrt{T+c})}\Big)^2
\nonumber
\\
&= -(W(X)+\nu)^2+\frac{X+c}{Z}
+\frac{1}{\sqrt{Z}}\int_1^{\Xi} \frac{dY \rho(Y)}{\sqrt{Y+c}}.
\label{MS-c}
\end{align}
This equation takes the form of 
(\ref{Int_eq_field1}) if we choose $\nu=0$, $Z=1$ 
and adjust\footnote{In \cite{Makeenko:1991ec}, $c$ is determined by
$c+\int_1^{\Xi} \frac{dY \rho(Y)}{\sqrt{Y+c}}=0$
from (\ref{MS-c}).  We are
  particularly interested in linearly spaced eigenvalues $e(x)=x$ where
  $\rho(Y)\propto \frac{1}{\sqrt{Y}}$, see (\ref{Int_eq_field1}).
  Then $\int_1^{\Xi} \frac{dY \rho(Y)}{\sqrt{Y+c}}$ diverges 
for $\Xi\to \infty$. This
  makes it necessary to normalise $W(1)=1$.} $c$ by 
\begin{align}
W(1)=1= \sqrt{1+c}
+ \frac{1}{2} \int_1^\Xi dT \frac{\rho(T)}{
(\sqrt{1+c} + \sqrt{T+c})\sqrt{T+c}}.
\label{c-lambda}
\end{align}
For $\rho(T)\sim T^{-\alpha}$ and $\alpha>0$, realised in our case, the
formula (\ref{MS-solution}) and the resulting condition on $c$
have a limit
$\Xi\to \infty$. 

Inserting $\rho(T)$ from (\ref{Int_eq_field1}) into (\ref{c-lambda})
we have an explicit expression of $\tilde{\lambda}^2$ in terms of $c$,
either with $c>-1$ real or
$c \in \mathbb{C}\setminus {]{-}\infty,-1]}$. Obviously, $c=0$ 
corresponds to $\tilde{\lambda}=0$. The 
implicit function theorem then provides a unique diffeomorphism 
$\tilde{\lambda}^2\mapsto c(\tilde{\lambda})$ on a neighbourhood 
of $0 \in \mathbb{R}$ or $0\in \mathbb{C}$. Since we will 
be able to express all correlation functions in terms of elementary
functions of $c(\tilde{\lambda},e)$ and $\rho(\tilde{\lambda},e)$, 
this proves analyticity of all correlation functions in these neighbourhoods.

\subsection{Linearly spaced eigenvalues of $E$}

The noncommutative field theory model of section \ref{sec:NCG}
translates to linearly spaced eigenvalues with $e(x)=x$. This yields
$X=(2x+1)^2$ and $\rho(Y)=\frac{2\tilde{\lambda}^2}{\sqrt{Y}}$.
The integral can be evaluated for $\Xi\to \infty$:
\begin{prop} \label{prop_MS}
Equation  (\ref{Int_eq_field1}) is for eigenvalue functions $e(x)=x$ 
and $Z=1,\nu=0$ solved by:
\begin{align}
W(X)&= \sqrt{X+c}
+ \frac{2\tilde{\lambda}^2}{\sqrt{X}}
\log \Big(\frac{(\sqrt{X+c}+\sqrt{X})(\sqrt{X}+1)}{
\sqrt{X}\sqrt{1+c}+\sqrt{X+c}}\Big)\;,
\\
1&=\sqrt{c+1} +2\tilde{\lambda}^2
\log\Big(1+\frac{1}{\sqrt{c+1}}\Big).
\label{W1=1}
\end{align}
\end{prop}
We thus get for the renormalised 1-point function
\begin{align}
\tilde{G}(x)&=\frac{1}{2\tilde{\lambda}}\Big(W((2x+1)^2)-(2x+1)\Big)
\label{Gx-final}
\\*
&= \frac{\sqrt{(2x{+}1)^2+c}-(2x{+}1)}{2\tilde{\lambda}}
+\frac{\tilde{\lambda}}{2x{+}1}
\log \Big(\frac{(2x{+}2)(\sqrt{(2x{+}1)^2+c}+2x{+}1)}{
(2x{+}1)\sqrt{1{+}c}+\sqrt{(2x{+}1)^2+c}}\Big),\nonumber
\end{align}
again with $c$ being the inverse solution of (\ref{W1=1}).

A numerical investigation shows that (\ref{W1=1}) has a
solution\footnote{In general, the critical value corresponds to 
$\rho_0:=1-\frac{1}{2} \int_1^\infty \frac{dZ
  \rho(Z)}{\sqrt{Z+c}^3}=0$. This function $\rho_0$ plays a key r\^ole in
higher correlation functions.} for
$-\tilde{\lambda}_c\leq \tilde{\lambda} \leq \tilde{\lambda}_c$ and
$\tilde{\lambda}_c =0.490686\dots$ attained at $c_c=-0.873759\dots$.
By choosing $c>0$ it is possible to simulate purely imaginary
$\tilde{\lambda}$. A perturbative solution of (\ref{W1=1})
gives as first terms
\begin{align}
c=-4\tilde{\lambda}^2 \log 2
-4\tilde{\lambda}^4(\log 2 -(\log 2)^2)
-2\tilde{\lambda}^6(2\log 2 -(\log 2)^2)
+
\mathcal{O}(\tilde{\lambda}^8).
\end{align}
This leads to the following series expansion of the renormalised 1-point
function:
\begin{align}
\tilde{G}(x)&=\frac{\tilde{\lambda}}{2x+1} \log(x+1)+
\tilde{\lambda}^3\Big(\frac{(\log(2))^2}{2x+1}-
\frac{(\log(2))^2}{(2x+1)^3}\Big)
\nonumber
\\
&+\tilde{\lambda}^5\Big(\frac{(\log(2))^2}{2x+1}+
\frac{2(\log(2))^3-(\log(2))^2}{(2x+1)^3}
-\frac{2(\log(2))^3}{(2x+1)^5}
\Big)
+\mathcal{O}(\tilde{\lambda}^7)\;.
\label{G1-O5}
\end{align}
It matches perfectly the Feynman graph computation 
(\ref{Feyn-G1}) of section~\ref{prelude}.

\smallskip

The scaling limit $\tilde{G}(x,y)=\lim\limits_{\mathcal{N},V\to\infty} 
\mu^2 G_{|V\mu^2x,V\mu^2y|}$ of
(\ref{SD-G2}) for the 2-point function is
\begin{align}
G(X,Y) =\tilde{G}(x(X),y(Y)) &= 2\frac{W(X)-W(Y)}{X-Y} \;.
\label{Gxy}
\end{align}
We refrain from spelling out the insertion of (\ref{Gx-final}). There
is no problem going to the diagonal: $\tilde{G}(x,x)=2W'(X)$.

\subsection{$N$-points functions}

According to (\ref{logZ}) the connected $(N{>}2)$-point functions are
\begin{align}
G_{|a_1  a_2  \dots  a_N|}
&= \frac{1}{V} \frac{\partial}{\partial J_{a_N a_1}}
\frac{\partial}{\partial J_{a_1 a_2}} \cdots
\frac{\partial}{\partial J_{a_{N-1} a_N}}
\log \frac{\mathcal{Z}[J]}{\mathcal{Z}[0]}
\Big|_{J=0} .
\end{align}
For pairwise different indices we compute, similarly to (\ref{2pointWard}),
\begin{align}
G_{|a_1\dots a_N|}&=
\frac{K}{\mathcal{Z}[0]}
\frac{\partial}{\partial J_{a_2 a_3}} \cdots
\frac{\partial}{\partial J_{a_{N} a_1}}
\exp \Big( {-}\frac{\lambda}{3V^2} \!\sum_{m,n,k}\!
\frac{\partial^3}{\partial J_{mn}\partial J_{nk}\partial J_{km}}
\Big) 
\Big( J_{a_2a_1} H_{a_1a_2}^{-1}\mathcal{Z}_{free}\Big)\Big|_{J=0}
\nonumber
\\
&=-\frac{\lambda}{V^2}
\frac{H_{a_1a_2}^{-1}}{\mathcal{Z}[0]}\sum_{m=0}^{\mathcal{N}}
\frac{\partial}{\partial J_{a_2 a_3}} \cdots
\frac{\partial}{\partial J_{a_{N}a_1}}
\frac{\partial}{\partial J_{a_1m}}
\frac{\partial}{\partial J_{ma_2}}
\mathcal{Z}[J]
\Big|_{J=0}
\nonumber
\\
&=-\frac{\lambda}{V}
\frac{H_{a_1a_2}^{-1}}{\mathcal{Z}[0]}\sum_{m=0}^{\mathcal{N}}
\frac{\partial}{\partial J_{a_2 a_3}} \cdots
\frac{\partial}{\partial J_{a_{N}a_1}}
\frac{\big(J_{m a_1}
\frac{\partial \mathcal{Z}[J]  }{\partial J_{ma_2}}
-J_{a_2m}
\frac{\partial \mathcal{Z}[J]  }{\partial J_{a_1m}}
\big)}{(E_{a_1}-E_{a_2})}
\Big|_{J=0}
\nonumber
\\
&=\lambda H_{a_1a_2}^{-1}
\frac{G_{|a_1a_3\dots a_{N}|}-G_{|a_2\dots a_{N}|}}{(E_{a_1}-E_{a_2})}
=\lambda
\frac{G_{|a_1a_3\dots a_{N}|}-G_{|a_2\dots a_{N}|}}{(E_{a_1}^2-E_{a_2}^2)}.
\label{NptWT}
\end{align}
The first line is the result of the
$\frac{\partial}{\partial J_{a_1a_2}}$ differentiation, and
in the step from the 2nd to 3rd line we have used
the Ward-Takahashi identity (\ref{WT}) for pairwise different
indices.  Formula (\ref{NptWT}) together with (\ref{SD-G2})
expresses $N$-point functions
recursively by factors $\frac{1}{(E_{a_i}^2-E_{a_j}^2)}$ and
$W_{|a_k|}$. We can solve this recursion:
\begin{prop} \label{prop_Npt_1}
The connected ($N{\geq} 2$)-point function is given
for pairwise different indices by
\begin{align}
G_{|a_1 a_2 , \dots  a_N|}
= \frac{\lambda^{N-2}}{2}
\sum_{k=1}^{N} W_{|a_k|}  \prod_{l=1, l\neq k}^{N} P_{a_ka_l}\;,\qquad
P_{ab} :=  \frac{1}{E_a^2-E_b^2}.
\label{GN}
\end{align}
\end{prop}
\begin{proof}
The formula is proved by induction, starting with $N=2$ which is
formula (\ref{SD-G2}) when inserting $P_{a_1a_2}=-P_{a_2a_1}$. Assume
it holds for $N$. Then using (\ref{NptWT}) and $P_{a_1a_2}=-P_{a_2a_1}$
we have
\begin{align*}
G_{|a_1 \dots  a_{N+1}|}
&= \lambda P_{a_1a_2}
( G_{|a_1a_3 \dots a_{N+1}|} - G_{|a_2 \dots a_{N+1}|})
\\
&=\frac{\lambda^{N-1}}{2} P_{a_1 a_2}
\bigg(
\sum_{k=1, k\neq 2}^{N+1} W_{|a_k|} \prod_{l=1, l\notin\{2,k\}}^{N+1} P_{a_ka_l}
- \sum_{k=2}^{N+1} W_{|a_k|}
 \prod_{l=2, l\neq k}^{N+1} P_{a_ka_l}
\bigg)
\notag
\\
&= \frac{\lambda^{N-1}}{2}
\bigg( W_{|a_1|}\prod_{l=2}^{N+1} P_{a_1a_l} +
W_{|a_2|}
\prod_{l=1,l\neq 2}^{N+1} P_{a_2a_l}
\nonumber
\\
&\qquad +
\sum_{k=3}^{N+1} W_{|a_k|} P_{a_1a_2}\Big(
P_{a_ka_1}\prod_{l=3, l\neq k}^{N+1} P_{a_ka_l}
-P_{a_ka_2}\prod_{l=3, l\neq k}^{N+1} P_{a_ka_l} \Big)
\bigg).
\end{align*}
Now the definition on $P_{a_ka_l}$ implies
\[
P_{a_1a_2} (P_{a_ka_1}-P_{a_ka_2})= P_{a_ka_1}P_{a_ka_2}\;,
\]
so that (\ref{GN}) follows for $N\mapsto N+1$.
\end{proof}

We can easily perform the scaling limit $\mathcal{N},V\to \infty$ to
functions
$\tilde{G}(x_1,\dots,x_N)=\lim_{\mathcal{N},V\to \infty}
\mu^{2(N-1)}G_{|V\mu^2x_1,\dots, V\mu^2x_N|}$
and $G(X_1,\dots,X_n):=\tilde{G}(x_1(X_1),\dots,x_N(X_N))$.
With $\lim (2E_{V\mu^2x_k}) =\mu^2 \sqrt{X_k}$
and thus $\lim (\mu^4 P_{kl})= \frac{4}{X_k-X_l}$ we have
\begin{align}
G(X_1,\dots,X_N)
=
\sum_{k=1}^{N} \frac{W(X_k)}{2\tilde{\lambda}}
\prod_{l=1, l\neq k}^{N} \frac{4\tilde{\lambda}}{X_k-X_l}.
\label{GXN}
\end{align}


\section{$N$-points function with
$B \ge 2$ boundaries}

\subsection{($N_1{+}\dots{+}N_B$)-point function
with one $N_i>1$}

To simplify notation let
$\frac{\partial^N}{\partial \mathbb{J}_{a_1 \dots a_N}}:=
\frac{\partial^N}{\partial J_{a_1 a_2} \dots
\partial J_{a_{N-1} a_N} \partial J_{a_N a_1}}$. We prove:
\begin{prop}
For $N_1>1$ one has
\begin{align}
G_{|a^1_1\dots a^1_{N_1}|\dots|a^B_1\dots a_{N_B}^B|}
= \lambda \frac{G_{|a^1_1a^1_3\dots a^1_{N_1}|a^2_1\dots a^2_{N_2}|\dots|a^B_1\dots a_{N_B}^B|}
-G_{|a_2^1a_3^1\dots a^1_{N_1}|a^2_1\dots a^2_{N_2}|\dots|a^B_1\dots a_{N_B}^B|}}{E_{a_1^1}^2-
E_{a_2^1}^2}\;.
\label{GNB}
\end{align}
\end{prop}
\begin{proof}
For pairwise different $a_i,b_j$ we have from (\ref{logZ})
\begin{align}
&G_{|a_1^1\dots a_{N_1}^1|\dots 
|a_1^B\dots a_{N_B}^B|}
= V^{B-2}
\frac{\partial^{N_1}}{\partial \mathbb{J}_{a_1^1 \dots a^1_{N_1}}}
\dots
\frac{\partial^{N_B}}{\partial \mathbb{J}_{a_1^B \dots a^B_{N_B}}}
\log\frac{\mathcal{Z}[J]}{\mathcal{Z}[0]}
\Big|_{J=0}
\nonumber
\\
&= V^{B-1}
\frac{\partial^{N_1-1}}{\partial J_{a^1_2 a^1_3}
\dots \partial J_{a^1_N a^1_1}}
\frac{\partial^{N_2}}{\partial \mathbb{J}_{a_1^2 \dots a^2_{N_2}}}
\dots
\frac{\partial^{N_B}}{\partial \mathbb{J}_{a_1^B \dots a^B_{N_B}}}
\bigg\{
\frac{K}{\mathcal{Z}[J]}
\nonumber
\\
& \qquad\qquad \times
\exp \Big( {-}\frac{\lambda}{3V^2} \!\sum_{m,n,k}\!
\frac{\partial^3}{\partial J_{mn}\partial J_{nk}\partial J_{km}}
\Big) 
\Big( J_{a_2^1a_1^1} H_{a_1^1a_2^1}^{-1}
\mathcal{Z}_{free}[J]\Big)\bigg\}\Big|_{J=0}
\tag{*}
\\
&= V^{B-3}\frac{(-\lambda)}{H_{a_1^1a_2^1}}
\frac{\partial^{N_1-1}}{\partial J_{a^1_2 a^1_3}
\dots \partial J_{a^1_N a^1_1}}
\frac{\partial^{N_2}}{\partial \mathbb{J}_{a^2_1 \dots a^2_{N_2}}}
\dots
\frac{\partial^{N_B}}{\partial \mathbb{J}_{a_1^B \dots a^B_{N_B}}}
\bigg\{
\frac{1}{\mathcal{Z}[J]}
\sum_{m=0}^{\mathcal{N}}
\frac{\partial^2 \mathcal{Z}}{\partial J_{a_1^1m}
\partial J_{ma_2^1}}
\bigg\}\Big|_{J=0}
 \notag
\\
&= V^{B-2}\frac{(-\lambda)}{E_{a_1^1}^2-E_{a_2^1}^2}
\frac{\partial^{N_1-1}}{\partial J_{a^1_2 a^1_3}
\dots \partial J_{a^1_N a^1_1}}
\frac{\partial^{N_2}}{\partial \mathbb{J}_{a^2_1 \dots a^2_{N_2}}}
\dots
\frac{\partial^{N_B}}{\partial \mathbb{J}_{a_1^B \dots a^B_{N_B}}}
\bigg\{
\frac{1}{\mathcal{Z}[J]}
\notag
\\*
&\qquad\qquad\qquad \times
\sum_{m=0}^{\mathcal{N}}
\Big(J_{ma_1^1}\frac{\partial\mathcal{Z}}{\partial J_{ma_2^1}}
- J_{a_2^1m}
\frac{\partial \mathcal{Z}}{\partial J_{a_1^1m}}
\Big)
\bigg\}\Big|_{J=0}
\tag{**}
\\
&=V^{B-2}\frac{(-\lambda)}{E_{a_1^1}^2-E_{a_2^1}^2}
\Big(\frac{\partial^{N_1-1}}{\partial \mathbb{J}_{a^1_2 \dots a^1_{N_1}}}
-\frac{\partial^{N_1-1}}{\partial \mathbb{J}_{a^1_1 a^1_3 \dots a^1_{N_1}}}
\Big)\frac{\partial^{N_2}}{\partial \mathbb{J}_{a^2_1 \dots a^2_{N_2}}}
\dots
\frac{\partial^{N_B}}{\partial \mathbb{J}_{a_1^B \dots a^B_{N_B}}}
\log \mathcal{Z}[J]
\Big|_{J=0}.
\tag{***}
\end{align}
Precisely for $N_1=2$ there is a surviving term of the
$J_{a_2^1a_1^1}$ differentiation, but the result cancels with
$\frac{K}{\mathcal{Z}[J]}$ so that further differentiations due to
$B\geq 2$ give zero. Therefore, all surviving differentiations of
$J_{a_2^1a_1^1}$ in (*) come from $\exp(-\frac{\lambda}{3V^2} \sum
\frac{\partial^3}{\partial J^3})$.  In (**) the
Ward-Takahashi identity (\ref{WT}) and $H_{ab}(E_a-E_b)=E_a^2-E_b^2$
are used.
Then $J_{ma_1^1}$ must be hit by $\frac{\partial}{\partial
  J_{a^1_{N_1}a_1^1}}$ and $J_{a_2^1m}$ by $\frac{\partial}{\partial
  J_{a^1_2a_3^1}}$, thus giving (***). The final line gives with
(\ref{logZ}) the assertion (\ref{GNB}).
\end{proof}

By symmetry in the boundary components we can recursively use
(\ref{GNB}) to express any ($N_1{+}\dots{+}N_B$)-point function
with one $N_i>1$ in terms of $G_{|a^1|a^2|\dots|a^B|}$. Since further
boundaries play a spectator r\^ole in (\ref{GNB}), we can easily adapt
the arguments of Proposition \ref{prop_Npt_1} to resolve this recursion:
\begin{prop} \label{prop_Npt_B}
Let $B\geq 2$. The connected ($N_1{+}\dots{+}N_B$)-point function
with one $N_i>1$ is given in terms of $P_{ab}:=\frac{1}{E_a^2-E_b^2}$ by
\begin{align}
&G_{|a^1_1\dots a^1_{N_1}|\dots|a^B_1\dots a_{N_B}^B|}
\label{GNB-final}
\\
&=\lambda^{N_1+\dots+N_B-B}
\sum_{k_1=1}^{N_1} \dots
\sum_{k_B=1}^{N_B}
G_{|a^1_{k_1}|\dots|a^B_{k_B}|}
\Big(\prod_{l_1=1, l_1\neq k_1}^{N_1} \!\!\!
P_{a^1_{k_1}a^1_{l_1}}\Big)\cdots
\Big(\prod_{l_B=1, l_B\neq k_B}^{N_B}\!\!\!
P_{a^B_{k_B}a^B_{l_B}}\Big),
\nonumber
\end{align}
its large-$(\mathcal{N},V)$ limit by
\begin{align}
&G(X^1_1,\dots, X^1_{N_1}|\dots|X^B_1,\dots, X_{N_B}^B)
\label{GNB-final-X}
\\
&=\tilde{\lambda}^{N_1+\dots+N_B-B}
\sum_{k_1=1}^{N_1} \dots
\sum_{k_B=1}^{N_B}
G(X^1_{k_1}|\dots|X^B_{k_B})
\prod_{\beta=1}^B \prod_{l_\beta=1, l_\beta\neq k_\beta}^{N_\beta} 
\frac{4}{X^\beta_{k_\beta}-X^\beta_{l_\beta}}.
\nonumber
\end{align}
\end{prop}

\subsection{SD-equation for 
$(1{+}\dots{+}1)$-point function}

\begin{prop}
Let $B\geq 2$. Then the $(1{+}\dots{+}1)$-point function satisfies
\begin{align}
W_{|a^1|}  G_{|a^1|a^2|\dots |a^B|}
&+\frac{\lambda^2}{V}
\sum_{m=0}^{\mathcal{N}} \frac{G_{|a^1|a^2|\dots |a^B|}
-G_{|m|a^2|\dots |a^B|}}{(E_{a^1}^2-E_m^2)}
\label{G1B}
\\
&=-\lambda
\sum_{\beta=2}^B G_{|a^1a^\beta a^\beta|a^2|
\stackrel{\beta}{\check{\dots\dots}}  |a^B|}
- \frac{\lambda}{V^2} G_{|a^1|a^1|a^2|\dots |a^B|}
\nonumber
\\
&
- \lambda\sum_{p=1}^{B-2} \sum_{2\leq i_1<\dots < i_p\leq B}
G_{|a^1|a^{i_1}|\dots |a^{i_p}|}  G_{|a^1|a^{j_1}|\dots |a^{j_{B-p-1}}|},
\nonumber
\end{align}
where $2 \leq j_1< \dots <j_{B-p-1} \leq B$ and
$\{i_1,\dots,i_p,j_1,\dots,j_{B-p-1}\}=\{2,\dots,B\}$,
and $\stackrel{\beta}{\check{\dots\dots}}$ denotes the omission of
$a^\beta$.
\end{prop}
\begin{proof}
We write down for pairwise different indices $a^\beta$ the formula
for the ($1{+}\dots{+}1$)-point
function in (\ref{logZ}) with $B\geq 2$ boundary components and perform
the $J_{a^1a^1}$-differentiation:
\begin{align}
&G_{|a^1|a^2|\dots |a^B|}
= V^{B-2}
\frac{\partial^B}{\partial J_{a^1a^1}\dots
J_{a^Ba^B}}\log\frac{\mathcal{Z}[J]}{\mathcal{Z}[0]}
\Big|_{J=0}
\nonumber
\\
&=\frac{ V^{B-1}\partial^{B-1}}{\partial J_{a^2a^2}\dots
J_{a^Ba^B}}\Big\{\frac{K}{\mathcal{Z}[J]}
\exp \Big( {-}\frac{\lambda}{3V^2} \!\sum_{m,n,k}\!
\frac{\partial^3}{\partial J_{mn}\partial J_{nk}\partial J_{km}}
\Big) 
\Big( (J_{a^1a^1}-\kappa)
H_{a^1a^1}^{-1}\mathcal{Z}_{free}[J]\Big)\Big\}\Big|_{J=0}
\nonumber
\\
&= V^{B-1}\frac{\partial^{B-1}}{\partial J_{a^2a^2}\dots
J_{a^Ba^B}}\Big\{\frac{1}{\mathcal{Z}[J]}
\frac{(-\lambda)}{V^2 H_{a^1a^1}}
\sum_{m=0}^{\mathcal{N}} \frac{\partial}{\partial J_{a^1m}}
\frac{\partial}{\partial J_{ma^1}} \mathcal{Z}
\Big\}\Big|_{J=0}
\nonumber
\\
&= V^{B-3}\frac{(-\lambda)}{H_{a^1a^1}}
\sum_{m=0}^{\mathcal{N}}
\frac{\partial^{B-1}}{\partial J_{a^2a^2}\dots J_{a^Ba^B}}
\Big\{
\frac{\partial^2 \log(\mathcal{Z})}{\partial J_{a^1m}
\partial J_{ma^1}}
+
\frac{\partial \log(\mathcal{Z})}{\partial J_{a^1m}}
\frac{\partial\log(\mathcal{Z})}{\partial J_{ma^1}}
\Big\}\Big|_{J=0}
\nonumber
\\
&= \frac{(-\lambda)}{H_{a^1a^1}}\Big\{
\frac{1}{V}\sum_{m=0}^{\mathcal{N}} G_{|a^1m|a^2|\dots |a^B|}
+\sum_{\beta=2}^B G_{|a^1a^\beta a^\beta|a^2| \stackrel{\beta}{\check{\dots\dots}}
  |a^B|}
+ \frac{1}{V^2} G_{|a^1|a^1|a^2|\dots |a^B|}
\nonumber
\\
&+ 2 G_{|a^1|}  G_{|a^1|a^2|\dots |a^B|}
+ \sum_{p=1}^{B-2} \sum_{2\leq i_1<\dots < i_p\leq B}
G_{|a^1|a^{i_1}|\dots |a^{i_p}|}  G_{|a^1|a^{j_1}|\dots |a^{j_{B-p-1}}|}\Big\},
\end{align}
with notations introduced in the proposition.
We multiply by $\frac{H_{aa}}{\lambda}$ and bring
$-2 G_{|a^1|}  G_{|a^1|a^2|\dots |a^B|}$ to the lhs, thus reconstructing
the function $W_{|a^1|}$ defined in (\ref{def:Wa}):
\begin{align}
\frac{1}{\lambda} W_{|a^1|}  G_{|a^1|a^2|\dots |a^B|}
+\frac{1}{V}\sum_{m=0}^{\mathcal{N}} G_{|a^1m|a^2|\dots |a^B|}
&=-\sum_{\beta=2}^B G_{|a^1a^\beta a^\beta|a^2|
\stackrel{\beta}{\check{\dots\dots}}  |a^B|}
- \frac{1}{V^2} G_{|a^1|a^1|a^2|\dots |a^B|}
\nonumber
\\
&
- \sum_{p=1}^{B-2} \sum_{2\leq i_1<\dots < i_p\leq B} \!\!\!
G_{|a^1|a^{i_1}|\dots |a^{i_p}|}  G_{|a^1|a^{j_1}|\dots |a^{j_{B-p-1}}|}.
\end{align}
Reducing the $(2{+}1{+}\dots{+}1)$-point function
by (\ref{GNB}) leads to the assertion (\ref{G1B}).
\end{proof}

Taking the scaling limit $
G(x^1|\dots|x^B):=\mu^{2(2-B)}
\lim_{\mathcal{N},V\to \infty}
G_{|V\mu^2x^1|\dots | V\mu^2x^B|}$, the
term $\frac{1}{V^2} G_{|a^1|a^1|a^2|\dots |a^B|}$ in (\ref{G1B}) goes
away, and we obtain a recursive system of affine equations for the
function with $B$ boundary components. To write these equations in
more condensed form, let us abbreviate for a set $I=\{i_1,\dots,i_p\}$
of indices
$G(X|Y\triangleleft^I):=G(X|Y^{i_1}|\dots|Y^{i_p})$.
With these notations, and including $\nu$ (here $=0$) from 
(\ref{MS-solution}) for later use in \cite{Grosse:2016qmk}, 
we can express the limit of (\ref{G1B})
in terms of $X^i:=(2e(x^i)+1)^2$ as follows:
\begin{align}
&(W(X^1)+\nu)  G(X^1|X\triangleleft^{\{2,\dots B\}})
+\frac{1}{2} \int_1^\Xi dT \rho(T)
\frac{G(X^1|X\triangleleft^{\{2,\dots B\}})-G(T|X\triangleleft^{\{2,\dots B\}})}{
(X-T)}
\nonumber
\\
&=-\tilde{\lambda}
\sum_{\beta=2}^B G\Big(X^1,X^\beta ,X^\beta|X\triangleleft^{\{2
\stackrel{\beta}{\check{\dots\dots}} B\}}\Big)
- \tilde{\lambda} \!\!\!
\sum_{{J\subset \{2,\dots,B\} \atop 1\leq |J| \leq B-2}} \!\!\!
G(X^1|X\triangleleft^J)
G(X^1|X\triangleleft^{\{2,\dots,B\}\setminus J}).
\label{G1B-lim}
\end{align}
The measure $\rho(T)$ was defined in (\ref{Int_eq_field1}). 
In presence of $\nu\neq 0$ we need a finite cut-off $\Xi$; the limit 
$\Xi\to \infty$ is only possible for the solutions. The
inhomogeneity only involves known functions with $<B$ boundary
components.

\subsection{Solution for the $(1{+}1)$-point function}

We specify the problem (\ref{G1B-lim}) to the
$1+1$-point function
\begin{align}
(W(X)+\nu) G(X|Y)
&= -\tilde{\lambda} G(X,Y,Y)
-\frac{1}{2} \int_1^\Xi dT \rho(T)\; \frac{G(X|Y)-G(T|Y)}{X-T}.
\label{int_eq_GY|X}
\end{align}
A perturbative solution of (\ref{int_eq_GY|X}) to
$\mathcal{O}(\tilde{\lambda}^4)$ suggests:
\begin{prop} The $(1{+}1)$-point function is given by 
\begin{align}
G(X|Y)&= \frac{4\tilde{\lambda}^2}{\sqrt{X+c}\cdot \sqrt{Y+c}\cdot
(\sqrt{X+c}+\sqrt{Y+c})^2},
\label{G1+1}
\end{align}
where $c(e,\tilde{\lambda})$ was defined in (\ref{c-lambda}).
\end{prop}
\begin{proof}
We insert the \emph{ansatz} (\ref{G1+1}) into the following integral:
\begin{align}
&-\frac{1}{2} \int_1^\Xi dT \rho(T)\; \frac{G(X|Y)-G(T|Y)}{X-T}
\nonumber
\\
&= -\frac{2 \tilde{\lambda}^2}{\sqrt{Y+c}}
\int_1^\Xi dT \rho(T)\; \frac{\frac{1}{\sqrt{X+c}\cdot
(\sqrt{X+c}+\sqrt{Y+c})^2}-\frac{1}{\sqrt{T+c}\cdot
(\sqrt{T+c}+\sqrt{Y+c})^2}}{X-T}
\nonumber
\\
&= \frac{2 \tilde{\lambda}^2}{\sqrt{X+c}\cdot \sqrt{Y+c}\cdot
(\sqrt{X+c}+\sqrt{Y+c})^2}
\int_1^\Xi \frac{dT \rho(T)}{
\sqrt{T+c} \cdot (\sqrt{X+c}+\sqrt{T+c})}
\nonumber
\\
& +\frac{2 \tilde{\lambda}^2}{\sqrt{Y{+}c}\cdot 
(\sqrt{X{+}c}+\sqrt{Y{+}c})^2}
\int_1^\Xi
\frac{dT \rho(T)}{\sqrt{T{+}c}}
\frac{(\sqrt{X{+}c}+\sqrt{T{+}c}+2\sqrt{Y{+}c})}{
(\sqrt{X{+}c}+\sqrt{T{+}c})
(\sqrt{T{+}c}+\sqrt{Y{+}c})^2}
\nonumber
\\
&= (W(X)+\nu) G(X|Y)
- \frac{4 \tilde{\lambda}^2}{\sqrt{Z} \sqrt{Y+c}\cdot
(\sqrt{X+c}+\sqrt{Y+c})^2}
\nonumber
\\
& -4 \tilde{\lambda}^2 \frac{\partial}{\partial Y}
\int_1^\Xi
\frac{dT \rho(T)}{\sqrt{T{+}c}}
\frac{1}{(\sqrt{X{+}c}+\sqrt{T{+}c})(\sqrt{X{+}c}+\sqrt{Y{+}c})
(\sqrt{Y{+}c}+\sqrt{T{+}c})}.
\label{intG1+1}
\end{align}
We have inserted the formula for $W$ from (\ref{MS-solution}). 
On the other hand, from (\ref{GXN}),
\begin{align}
&G(X,Y,Y)=8\tilde{\lambda} \lim_{Y_1\to Y} \frac{\frac{W(X)-W(Y)}{(X-Y)}
-\frac{W(Y_1)-W(Y)}{(Y_1-Y)}}{X-Y_1}
=8\tilde{\lambda} \frac{\partial}{\partial Y} \frac{W(X)-W(Y)}{X-Y}
\nonumber
\\
&= 8\tilde{\lambda} \frac{\partial}{\partial Y}
\Big\{ \frac{1}{\sqrt{Z}(\sqrt{X+c}+\sqrt{Y+c})}
\nonumber
\\
&- \frac{1}{2} \int_1^\Xi \frac{dT\rho(T)}{\sqrt{T{+}c}}
\frac{1}{(\sqrt{X{+}c}+\sqrt{T{+}c})(\sqrt{X{+}c}+\sqrt{Y{+}c})
(\sqrt{Y{+}c}+\sqrt{T{+}c})}
\Big\}.
\end{align}
Adding $(-\tilde{\lambda})G(X,Y,Y)$ to (\ref{intG1+1})
yields $(W(X)+\nu)G(X|Y)$, as required by (\ref{int_eq_GY|X}).
\end{proof}
Note that (\ref{G1+1}) is essentially the same as 
\cite[eq.\ (93)]{Grosse:2005ig}.

\subsection{Solution for the $(1{+}1{+}1)$-point function}

We specify the problem (\ref{G1B-lim}) to the
$(1{+}1{+}1)$-point function
\begin{align}
(W(X)+\nu) &G(X|Y^2|Y^3)
+\frac{1}{2} \int_1^\Xi dT \rho(T)\; \frac{G(X|Y^2|Y^3)-G(T|Y^2|Y^3)}{X-T}
\nonumber
\\
&= -\tilde{\lambda} G(X,Y^2,Y^2|Y^3)-\tilde{\lambda} G(X,Y^3,Y^3|Y^2)
-2\tilde{\lambda} G(X|Y^2) G(X|Y^3). \label{int_eq_G(X|Y|Z}
\end{align}
We have with (\ref{GNB-final-X})
\begin{align}
& G(X,Y^2,Y^2|Y^3)=16\tilde{\lambda}^2 
\frac{\partial}{\partial Y^2} \frac{G(X|Y^3)-G(Y^2|Y^3)}{X-Y^2}
\nonumber
\\
&= -128\tilde{\lambda}^4 \frac{\partial^2}{\partial Y^2 \partial Y^3}
\Big\{\frac{\frac{1}{\sqrt{X+c}(\sqrt{X+c}+\sqrt{Y^3+c})}
-\frac{1}{\sqrt{Y^2+c}(\sqrt{Y^2+c}+\sqrt{Y^3+c})}
}{X-Y^2}\Big\}
\nonumber
\\
&=  \frac{\partial^2}{\partial Y^2 \partial Y^3}
\Big\{\frac{128\tilde{\lambda}^4(\sqrt{X{+}c}+\sqrt{Y^2{+}c}+\sqrt{Y^3{+}c})}{
\sqrt{X{+}c}\sqrt{Y^2{+}c}(\sqrt{X{+}c}+\sqrt{Y^3{+}c})
(\sqrt{X{+}c}+\sqrt{Y^2{+}c})
(\sqrt{Y^2{+}c}+\sqrt{Y^3{+}c})}\Big\}
\nonumber
\end{align}
and consequently
\begin{align}
& G(X,Y^2,Y^2|Y^3)+G(X,Y^3,Y^3|Y^2)+2G(X|Y^2)G(X|Y^3)
\nonumber
\\
&=
\frac{\partial^2}{\partial Y^2 \partial Y^3}
\Big\{\frac{128\tilde{\lambda}^4(\sqrt{X+c}+\sqrt{Y^2+c}+\sqrt{Y^3+c})}{
\sqrt{X+c}\sqrt{Y^2+c}\sqrt{Y^3+c}(\sqrt{X+c}+\sqrt{Y^3+c})
(\sqrt{X+c}+\sqrt{Y^2+c})}
\nonumber
\\
&+\frac{128\tilde{\lambda}^4 }{
\sqrt{X+c}^2(\sqrt{X+c}+\sqrt{Y^3+c})
(\sqrt{X+c}+\sqrt{Y^2+c})}\Big\}
\nonumber
\\
&= \frac{\partial^2}{\partial Y^2 \partial Y^3}
\Big\{\frac{128\tilde{\lambda}^4}{\sqrt{X+c}^2\sqrt{Y^2+c}\sqrt{Y^3+c}}\Big\}
=
\frac{32\tilde{\lambda}^4}{\sqrt{X+c}^2\sqrt{Y^2+c}^3\sqrt{Y^3+c}^3}.
\end{align}
Because of the factorisation the only reasonable ansatz is
\begin{align}
G(X|Y^2|Y^3)=\frac{(-32) \gamma \tilde{\lambda}^5}{
\sqrt{X+c}^3\sqrt{Y^2+c}^3\sqrt{Y^3+c}^3}.
\label{G1+1+1}
\end{align}
This gives as prefactor of 
$\frac{-32 \tilde{\lambda}^5}{\sqrt{Y^2+c}^3\sqrt{Y^3+c}^3}$ in 
(\ref{int_eq_G(X|Y|Z}) (with exchanged lhs and rhs 
and use of (\ref{MS-solution})):
\begin{align}
\frac{1}{X+c}&=\frac{\gamma}{\sqrt{Z}(X+c)}+\frac{\gamma}{2}
\int_1^\Xi \frac{dT \rho(T)}{\sqrt{(T+c)}\sqrt{(X+c)}^3
(\sqrt{X+c}+\sqrt{T+c})}
\nonumber
\\
&+ \frac{\gamma}{2}
\int_1^\Xi dT \rho(T) \frac{\frac{1}{\sqrt{(X+c)}^3}-
\frac{1}{\sqrt{(T+c)}^3}}{X-T}
\nonumber
\\
&=\frac{\gamma}{\sqrt{Z}(X+c)}-\frac{\gamma}{2(X+c)}
\int_1^\Xi \frac{dT \rho(T)}{\sqrt{(T+c)}^3}
\nonumber
\\
\Rightarrow\quad \gamma
&=\frac{1}{ \rho_0} ,\qquad \rho_0 := \frac{1}{\sqrt{Z}}-
\int_1^\Xi \frac{dT \rho(T)}{2\sqrt{(T+c)}^3}.
\label{gamma-3empty}
\end{align}
For linearly spaced eigenvalues $e(x)=x$ and $Z=1$, i.e.\ 
$\rho(T)=\frac{2\tilde{\lambda}^2}{\sqrt{T}}$, this
amounts to
$\rho_0 = 1-\frac{2\tilde{\lambda}^2}{\sqrt{1+c}(\sqrt{1+c}+1)}$.

\subsection{Solution for the $(1{+}\dots{+}1)$-point
 function for $B\geq 4$}

This is the most elaborate section of the paper. Over the next 6 pages
we prepare the proof of Theorem~\ref{thm:G1B}. Eq.\ 
(\ref{G1+1+1}) suggests that
all $(1{+}\dots{+}1)$-point functions with $B\geq 3$ factorise. We
 make the ansatz
\begin{align}
G(X^1|\dots|X^B)
&=
\frac{(-2\tilde{\lambda})^{3B-4}}{\rho_0}
\sum_{M=0}^{B-3} \gamma^M_B
\frac{d^{M}}{dt^{M}}\sqrt{X+c-2t}^{-3}_{\{1,\dots,B\}}\Big|_{t=0},
\label{goodansatz}
\\
\text{where } &\qquad
\sqrt{X+c-2t}^{-3}_{I}
:= \prod_{\beta \in I} \frac{1}{\sqrt{X^\beta+c-2t}^3}.
\nonumber
\end{align}
Our aim is to compute the coefficients $\gamma^M_B$ starting with 
$\gamma^M_3=\delta_{M,0}$. 
\begin{lem}
Assume (\ref{goodansatz}). Then
\begin{align}
&(W(X^1)+\nu)G(X^1|X\triangleleft^{\{2,\dots,B\}})
+\frac{1}{2} \int_1^\Xi dT \rho(T)
\frac{G(X^1|X\triangleleft^{\{2,\dots,B\}})
-G(T|X\triangleleft^{\{2,\dots,B\}})}{X^1-T}
\nonumber
\\
&=
\frac{(-2\tilde{\lambda})^{3B-4}}{\rho_0} \sum_{M=0}^{B-3} \gamma_B^M
\sum_{j=0}^M \binom{M}{j}
\sum_{l=0}^{j}
\frac{(2j{+}1)!! \rho_{j-l}}{\sqrt{X^1+c}^{2l+2}}
\frac{d^{M-j}}{dt^{M-j}} \sqrt{X{+}c-2t}_{\{2,\dots,B\}}^{-3}\Big|_{t=0},
\label{W-ansatz-1B}
\\
&\text{where}\qquad\qquad
\rho_l:=\frac{\delta_{l,0}}{\sqrt{Z}} -\frac{1}{2}
\int_1^\Xi \frac{dT \rho(T)}{\sqrt{T+c}^{3+2l}}.
\label{def-rhol}
\end{align}
\end{lem}
\begin{proof}
We distribute the $t$-derivatives by Leibniz rule.
The prefactor of
$\frac{(-2\tilde{\lambda})^{3B-4}}{\rho_0} \gamma_B^M (2j+1)!!\binom{M}{j}
\frac{d^{M-j}}{dt^{M-j}} \sqrt{X+c-2t}^{-3}_{\{2,\dots,B\}}\big|_{t=0}$
under the sum over $j,M$ is
\begin{align}
&(W(X^1)+\nu) \frac{1}{\sqrt{X^1+c}^{3+2j}}
+\frac{1}{2} \int_1^\Xi dT \rho(T)
\frac{\frac{1}{\sqrt{X^1+c}^{3+2j}}-\frac{1}{\sqrt{T+c}^{3+2j}}}{X^1-T}
\nonumber
\\
&= \frac{1}{\sqrt{Z}\sqrt{X^1+c}^{2j+2}}
-\frac{1}{2} \int_1^\Xi dT \rho(T)
\frac{\sum_{l=1}^{2j+2}
\sqrt{X^1+c}^l \sqrt{T+c}^{2j+2-l}}{
\sqrt{X^1+c}^{3+2j}\sqrt{T+c}^{3+2j}
(\sqrt{X^1+c}+\sqrt{T+c})
}
\nonumber
\\
&= \frac{1}{\sqrt{Z}\sqrt{X^1+c}^{2j+2}}
-\frac{1}{2} \sum_{l=0}^{j}
\frac{1}{\sqrt{X^1+c}^{2(j-l)+2}}
\int_1^\Xi \frac{dT \rho(T)}{\sqrt{T+c}^{3+2l}}
= \sum_{l=0}^{j} \frac{\rho_l}{\sqrt{X^1+c}^{2(j-l)+2}},
\end{align}
with $\rho_l$ defined in (\ref{def-rhol}). The step from the first to
second line relies on
\begin{align}
\frac{\frac{1}{\sqrt{X^1+c}^{3+2j}}-
\frac{1}{\sqrt{X^2+c}^{3+2j}}}{X^1-X^2}
= -
\sum_{l=0}^{2j+2} \frac{\sqrt{X^1+c}^{l}\sqrt{X^2+c}^{2j+2-l}}{
\sqrt{X^1+c}^{3+2j}
\sqrt{X^2+c}^{3+2j}
(\sqrt{X^1+c}+
\sqrt{X^2+c})}\label{quot-sqrt}
\end{align}
with compensation of $l=0$ with the integral in $W(X^1)$ according to
(\ref{MS-solution}). After a reflection $l\mapsto j-l$
we arrive at (\ref{W-ansatz-1B}).
\end{proof}

\begin{lem}
Assume (\ref{goodansatz}). Then the first term on the rhs of (\ref{G1B-lim})
and the $|J|=1$ and $|J|=B {-}2$ contributions to the last term 
combine to
\begin{align}
&-\tilde{\lambda}
\sum_{\beta=2}^B
\Big(G(X^1,X^\beta,X^\beta|
X\triangleleft^{\{2,\stackrel{\beta}{\check{\dots\dots}},B\}})+
2G(X^1|X^\beta)G(X^1|X\triangleleft^{\{2,\stackrel{\beta}{\check{\dots\dots}},B\}})
\Big)
\nonumber
\\*
&=\frac{(-2\tilde{\lambda})^{3B-4} }{\rho_0}
\sum_{M=0}^{B-4}   \sum_{j=0}^{M} \sum_{l=0}^{j+1}
\frac{\gamma_{B-1}^{M}}{
\sqrt{X^1{+}c}^{4+2j-2l}}
\binom{M}{j}
\frac{(2j+1)!!(2l+1)}{(2l+1)!!}
\nonumber
\\*
& \times
\sum_{\beta=2}^B
\Big(\frac{d^l}{dt^l} \frac{1}{\sqrt{X^\beta+c-2t}^3}\Big)
\Big(\frac{d^{M-j}}{dt^{M-j}}
\sqrt{X{+}c{-}2t}_{\{2,\stackrel{\beta}{\check{\dots}},B\}}^{-3}\Big)
\Big|_{t=0}.
\label{G3+1_2+B}
\end{align}
\end{lem}
\begin{proof}
It suffices to take $\beta=2$ and then to permute. From 
(\ref{GNB-final-X}) we have
\begin{align}
& G(X^1,X^2,X^2|X\triangleleft^{\{3,\dots,B\}})
=16\tilde{\lambda}^2
\frac{\partial}{\partial X^2} \frac{G(X^1|X\triangleleft^{\{3,\dots,B\}})
-G(X^2|X\triangleleft^{\{3,\dots,B\}})}{X^1-X^2}.
\end{align}
We insert (\ref{goodansatz}) for $B\mapsto B-1$. With Leibniz rule and
(\ref{quot-sqrt}) one has
\begin{align}
& -\tilde{\lambda} G(X^1,X^2,X^2|X\triangleleft^{\{3,\dots,B\}})
\nonumber
\\
&=
\sum_{M=0}^{B-4} \sum_{j=0}^{M} \binom{M}{j} \sum_{l=0}^{2j+2}
\frac{\partial}{\partial X^2}
\Big\{\frac{\gamma_{B-1}^{M} (2j+1)!!
\sqrt{X^1+c}^{l}
\sqrt{X^2+c}^{2j+3-(l+1)}}{
\sqrt{X^1+c}^{3+2j}
\sqrt{X^2+c}^{3+2j}
(\sqrt{X^1+c}+
\sqrt{X^2+c})}\Big\}
\nonumber
\\
& \times 16\tilde{\lambda}^3 \cdot \frac{(-2 \tilde{\lambda})^{3B-7}}{\rho_0}
\frac{d^{M-j}}{dt^{M-j}} \sqrt{X{+}c-2t}_{\{3,\dots,B\}}^{-3}\Big|_{t=0}.
\label{G3+1B}
\end{align}
The other term reads  with (\ref{G1+1}) as well as
(\ref{goodansatz}) for $B\mapsto B-1$
\begin{align}
& -2\tilde{\lambda} G(X^1|X^2)G(X^1|X\triangleleft^{\{3,\dots,B\}})
\label{G1+1G1B}
\\
&=
\sum_{M=0}^{B-4} \sum_{j=0}^{M} \binom{M}{j}
\frac{(2j+1)!!\gamma_{B-1}^M}{\rho_0 
 \sqrt{X^1+c}^{3+2j}}
\frac{\partial}{\partial X^2} \Big\{
\frac{(-2\tilde{\lambda})\cdot (-8\tilde{\lambda}^2)\cdot 
(-2\tilde{\lambda})^{3B-7}}{
\sqrt{X^1{+}c}(\sqrt{X^1{+}c}+\sqrt{X^2{+}c})} \Big\} 
\nonumber
\\
& \times \frac{d^{M-j}}{dt^{M-j}} \sqrt{X{+}c-2t}_{\{3,\dots,B\}}^{-3}\Big|_{t=0}.
\nonumber
\end{align}
Bringing (\ref{G3+1B})+(\ref{G1+1G1B}) to common $X^1$-$X^2$ denominator
$\frac{1}{
\sqrt{X^1+c}^{4+2j}\sqrt{X^2+c}^{3+2j}}$
(before $X^2$-differentiation) produces
a total numerator
\[\sum_{l=0}^{2j+3}
\sqrt{X^1{+}c}^{l}\sqrt{X^2{+}c}^{3+2j-l}
=\sum_{l=0}^{j+1}(\sqrt{X^1{+}c}+
\sqrt{X^2{+}c}) \sqrt{X^1{+}c}^{2l}\sqrt{X^2{+}c}^{2j+2-2l}.
\]
After cancellation and differentiation with respect to $X^2$ we have
\begin{align}
& -\tilde{\lambda}\big(G(X^1,X^2,X^2|X\triangleleft^{\{3,\dots,B\}})
+ 2G(X^1|X^2)G(X^1|X\triangleleft^{\{3,\dots,B\}})\big)
\\
&= \frac{(-2\tilde{\lambda})^{3B-4}}{\rho_0} 
\sum_{M=0}^{B-4} \sum_{j=0}^{M} \binom{M}{j}
\gamma^M_{B-1}
\sum_{l=0}^{j+1}
\frac{(2j+1)!! (1+2l)}{
\sqrt{X^1{+}c}^{4+2j-2l}\sqrt{X^2{+}c}^{2l+3}}
\nonumber
\\
&\times \frac{d^{M-j}}{dt^{M-j}}
\sqrt{X{+}c-2t}_{\{3,\dots,B\}}^{-3}\Big|_{t=0}.
\nonumber
\end{align}
We write $\frac{1}{\sqrt{X^2+c}^{2l+3}}=\frac{1}{(2l+1)!!}
\frac{d^l}{dt^l} \frac{1}{\sqrt{X^2{+}c-2t}^3}$, repeat these steps for
all $X^{\beta\geq 2}$ and sum over $\beta$, thus establishing
the formula.
\end{proof}

The remaining terms with $2\leq |J|\leq B-3$ in the last term of 
(\ref{G1B-lim}) are straightforward:
\begin{align}
&
-\tilde{\lambda}
\sum_{J \subset \{2,\dots,B\},|J|=p} (G(X^1|X\triangleleft^{J})
G(X^1|X\triangleleft^{ \{2,\dots,B\}\setminus J})
\label{GkGB-k}
\\
&= \frac{1}{2}\cdot
\frac{(-2\tilde{\lambda})^{3B-4}}{\rho_0^2}
\sum_{M'=0}^{p+1-3}\sum_{M''=0}^{B-p-3}
\sum_{j'=0}^{M'}\sum_{j''=0}^{M''} \binom{M'}{j'}\binom{M''}{j''}
\frac{(2j'{+}1)!!(2j''{+}1)!! \gamma^{M'}_{p+1}
\gamma^{M''}_{B-p}}{
\sqrt{X^1+c}^{6+2j'+2j''}}
\nonumber
\\
&\qquad \times
\sum_{J \subset \{2,\dots,B\},|J|=p}
\Big(\frac{d^{M'-j'}}{dt^{M'-j'}}
\sqrt{X{+}c-2t}^{-3}_{J}\Big)
\Big(\frac{d^{M''-j''}}{dt^{M''-j''}}
\sqrt{X{+}c{-}2t}^{-3}_{\{2,\dots,B\}\setminus J}\Big)\Big|_{t=0}.
\nonumber
\end{align}

Symbolically we are left with the problem
$\big[\eqref{W-ansatz-1B}
=\eqref{G3+1_2+B}
+ \sum_{p=2}^{B-3}\eqref{GkGB-k}\big]$
to be solved for $\gamma_B^M$, provided the ansatz is
consistent. By shifting indices we select the common coefficient of 
$\frac{(-2\tilde{\lambda})^{3B-4}}{\rho_0 \sqrt{X_1+c}^{6+2l}} 
\prod_{\beta=2}^B \big(\frac{1}{m_\beta!}\frac{d^{m_\beta}}{dt^{m_\beta}}
\frac{1}{\sqrt{X^\beta+c-2t}^{3}} \big)\big|_{t=0}$ 
in this equation:
\begin{lem}
Assume (\ref{goodansatz}). Then (\ref{G1B-lim}) amounts to the
following system of equations for integers $l\geq -2$ and 
$(B{-}1)$-tuples $\mathcal{M}=(m_2,\dots,m_B)$ with
$M:=m_2+\dots+m_B$:
\begin{align}
&\sum_{j=0}^{B-5-M-l} (M+2+l+j)!\gamma_B^{M+2+l+j}
\frac{(2j{+}2l{+}5)!! \rho_{j}}{(l+2+j)!}
\label{eq:gammaMB-all}
\\
&=
(M+l+1)!\gamma_{B-1}^{M+l+1}
\sum_{\beta=2}^B \frac{(2l+2m_\beta+3)!!(2m_\beta+1)m_\beta!}{
(l+m_\beta+1)!(2m_\beta+1)!!}
\nonumber
\\
&+\frac{1}{2\rho_0}\sum_{l'+l''=l}
\frac{(2l'{+}1)!!(2l''{+}1)!! }{l'!l''!} \!\!\!
\sum_{\mathcal{M}'\cup \mathcal{M}''=\mathcal{M}} \!\!
(M'+l')!\gamma^{M'+l'}_{\#(\mathcal{M}')+1}
(M''+l'')!\gamma^{M''+l''}_{\#(\mathcal{M}'')+1}.
\nonumber
\end{align}
The sum in the last line (which contributes only for $l\geq 0$) 
is over all partitions of 
$\mathcal{M}$ into two subtuples $\mathcal{M}',
\mathcal{M}''$ of $\#(\mathcal{M}')$ and 
$\#(\mathcal{M}'')$ elements which sum up to $M'$ and $M''$, 
respectively. The initial condition is $\gamma^M_3=\delta_{M,0}$.
\end{lem}

For the solution we have to introduce:
\begin{defn}
The Bell polynomials\footnote{For an overview about Bell polynomials, see 
\url{https://en.wikipedia.org/wiki/Bell_polynomials} or 
\url{https://www.encyclopediaofmath.org/index.php/Bell_polynomial}. 
Many identities are proved in \cite{Birmajer} and references therein.} $B_{n,k}$ are defined by $B_{0,k}(\{~\})=\delta_{k,0}$ and 
$B_{n,k}(\{x_j\}_{j=1}^{n-k+1})
=\sum \frac{n!}{j_{1}!j_{2}!\cdots j_{n-k+1}!}
(\frac{x_{1}}{1!})^{j_{1}}
(\frac{x_{2}}{2!})^{j_{2}}\cdots (\frac{x_{n-k+1}}{
(n-k+1)!})^{j_{n-k+1}}$,
for $n{\geq} 1$, where the sum is over non-negative integers 
$j_1,\dots, j_{n-k+1}$
with $j_1 +j_2+\dots+  j_{n-k+1}=k$ and
$1j_1 +2j_2+\dots+  (n-k-1)j_{n-k+1}=n$.
\end{defn}
\begin{lem}
The Bell polynomials satisfy the identity
\begin{align}
\sum_{j=1}^{n-k}
(\alpha j{+}\beta ) \binom{n}{j} x_j B_{n-j,k}\big(x_1,\dots,x_{n-j-k+1}\big)
= (\alpha n{+}\beta(k{+}1)) B_{n,k+1} \big(x_1,\dots,x_{n-k}\big) .
\label{Bell-1}
\end{align}
\end{lem}
\begin{proof}
This follows from \cite[Lemma 8]{Birmajer},
\begin{align*}
\binom{n}{m} B_{m,l}(\{x\})B_{n-m,k-l}(\{x\})
=\sum_{v\in \pi(n,k)  } \frac{n!}{v_1!v_2!\cdots}
W_{m,l}(v) \Big(\frac{x_1}{1!}\Big)^{v_1}\Big(\frac{x_2}{2!}\Big)^{v_2}
\cdots,
\end{align*}
where the $\pi(n,k)$ is the set of $v_1,v_2,\dots \geq 0$ with
$1v_1+2v_2+\dots =n$ and $v_1+v_2+\dots =k$.
We only need $l=1$ where the general definition of
$W_{m,l}(v)$ given in \cite[eq.\ (2)]{Birmajer}
reduces to $W_{m,1}(v)=v_m$. Moreover, $B_{m,1}(\{x\})=x_m$. Therefore,
\begin{align*}
&\sum_{m=1}^{n-k+1} \!\! (\alpha m{+}\beta)\binom{n}{m} x_m B_{n-m,k-1}(\{x\})
=\sum_{m=1}^{n-k+1} \!\! (\alpha m{+}\beta)\binom{n}{m}
B_{m,1}(\{x\})B_{n-m,k-1}(\{x\})
\nonumber
\\
&=\sum_{v\in \pi(n,k)  } \frac{n!}{v_1!v_2!\cdots}
\Big(\sum_{m=1}^{n-k+1} \!\! (\alpha mv_m{+}\beta v_m)\Big)
\Big(\frac{x_1}{1!}\Big)^{v_1}\Big(\frac{x_2}{2!}\Big)^{v_2} \cdots
=(\alpha n+\beta k) B_{n,k}(\{x\}).
\end{align*}
A shift in $k$ yields the result.
\end{proof}

\begin{prop}
The solution of (\ref{eq:gammaMB-all}) for $l=-2$ and $l=-1$, 
\begin{align}
\sum_{j=0}^{B-3-M} \binom{M+j}{j} (2j+1)!!
\rho_j \gamma^{M+j}_B
&=\gamma^{M-1}_{B-1} , \label{eq:gammaMB}
\\
\sum_{j=0}^{B-4-M} \binom{M+1+j}{j+1}  (2j+3)!!
\rho_j  \gamma^{M+1+j}_B
&= (2M+B-1) \gamma^M_{B-1},
\label{eq:gammaMB-0}
\end{align}
where $M\in \{0,\dots,B-3\}$ and under initial condition 
$\gamma_3^M=\delta_{M,0}$, is
\begin{align}
\gamma^M_B= \frac{1}{\rho_0^{B-3}} \sum_{K=0}^{B-3-M}
\frac{(B-3+K)!}{(B-3-M)!M!} B_{B-3-M,K}\Big(
\Big\{-\frac{(2r+1)!!\rho_r}{(r+1)\rho_0} \Big\}_{r=1}^{B-2-M-K}\Big).
\label{solution:gammaMB}
\end{align}
\end{prop}
\begin{proof}
We start with (\ref{eq:gammaMB}). 
The formula correctly captures the case $M=B-3$ where only $j=0$
contributes in (\ref{eq:gammaMB}), giving the solution
$\gamma^{B-3}_B=\frac{1}{\rho_0^{B-3}}$. We proceed by twofold
induction in increasing $B$ and increasing $s:=B-3-M$.
We rearrange (\ref{eq:gammaMB}) as an equation for $\gamma^M_B$.
All other terms either have less $B$ (namely $\gamma^{M-1}_{B-1}$)
or less $s$ (namely $\gamma^{M+j}_{B}$, $j\geq 1$) so that the induction
hypothesis applies. We have with
$x_r:=-\frac{(2r+1)!!\rho_r}{(r+1)\rho_0} $
and $s:=B-3-M$ in (\ref{eq:gammaMB}):
\begin{align}
\gamma^M_B &= \frac{1}{\rho_0^{B-3}} \sum_{K=0}^{s}
\frac{(B-4+K)!}{s!(M-1)!} B_{s,K}\big(
\{x_r\}_{r=1}^{s-K+1}\big)
\tag{*}
\\
&+ \frac{1}{\rho_0^{B-3}}
\sum_{j=1}^{s}
 \binom{M{+}j}{j} (j+1) x_j
\sum_{K=0}^{s-j}
\frac{(B-3+K)!}{(s{-}j)!(M{+}j)!} B_{s-j,K}\big(
\{x_r\}_{r=1}^{s-j-K+1}\big).
\tag{**}
\end{align}
We exchange the summation order
$\sum_{j=1}^{s}
\sum_{K=0}^{s-j}= \sum_{K=0}^{s-1} \sum_{j=1}^{s-K}$ and use
(\ref{Bell-1}) to express the last line as
\begin{align*}
(**)=\frac{1}{\rho_0^{B-3}} \sum_{K=0}^{s-1}
\frac{(B-3+K)!}{s!M!}
(s+K+1)
B_{s,K+1}\big(
\{x_r\}_{r=1}^{s-K}\big).
\end{align*}
We shift the index $K+1\mapsto K$ and redistribute the resulting
$(s+K)=(B-3+K)-M$: Its part $(-M)$ cancels the rhs of the first line (*),
and $(B-3+K)$ increases the factorial to the claimed formula
(\ref{solution:gammaMB}).

\smallskip 

We check consistency with (\ref{eq:gammaMB-0}).
For $M=B-4$ the lhs restricts to $j=0$, and both sides evaluate to 
$\frac{3(M+1)}{\rho_0^{B-4}}$. For $M \geq B-5$
we express the lhs in terms of $x_r:=-\frac{(2r+1)!!}{(r+1)}
\frac{\rho_r}{\rho_0}$ and insert 
(\ref{solution:gammaMB}). Then the $j\geq 1$ part of the lhs 
becomes after exchanging the $K$-$j$ summation
\begin{align*}
\eqref{eq:gammaMB-0}^{\text{lhs}}_{j\geq 1} 
&= {-}\!\!\!\!\! \sum_{K=0}^{B-5-M} 
 \sum_{j=1}^{B-4-M-K} \!\!\!\!\! \frac{(B{-}3{+}K)!}{(B{-}4{-}M)!M!}
\binom{B{-}4{-}M}{j}
(2j{+}3) x_j B_{B{-}4{-}M{-}j,K}(\{x_r\})
\\
&= {-}\!\!\!\!\!
\sum_{K=0}^{B-5-M} \!\!\! 
\frac{(B{-}3{+}K)!}{(B{-}4{-}M)!M!}
\underbrace{(2(B{-}M{-}4)+3(K{+}1))}_{=3(B-2+K) -(B+2M-1)}
B_{B-4-M,K+1}(\{x_r\}),
\end{align*}
where (\ref{Bell-1}) has been used for $\alpha=2$, $\beta=3$.
Its part $3(B-2+K)$, after a shift $K+1\mapsto K$, evaluates to 
$-3(M+1)\gamma^{M+1}_B$. The remainder gives 
$\frac{(2M+B-1)}{\rho_0} \gamma^M_{B-1}$, so that 
(\ref{eq:gammaMB-0}) is true.
\end{proof}

Remains (\ref{eq:gammaMB-all}) for $l\geq 0$. Because 
of permutation symmetry we can assume $\mathcal{M}
=(\underbrace{0,\dots,0}_{n_0} ,\dots ,
\underbrace{p,\dots,p}_{n_p})$ 
with $n_0+\dots+n_p=B-1=:N$ and 
$0n_0+1n_1+\dots+pn_p= M$. Then the sum over subtupels 
$\mathcal{M}'
=(\underbrace{0,\dots,0}_{n_0'} ,\dots ,
\underbrace{p,\dots,p}_{n_p'})$ amounts to the sum over 
$0\leq n_i'\leq n_i$ with multiplicity 
$\binom{n_0}{n_0'}\cdots \binom{n_p}{n_p'}$. Therefore
(\ref{solution:gammaMB}) solves 
(\ref{eq:gammaMB-all}) for $l\geq 0$ iff the following is true:
\begin{conj}
\label{conj}
For any $l,n_0,\dots,n_p \in \mathbb{N}$,
the Bell polynomials satisfy the identity 
(with 
$n_0+\dots+n_p=N$ and 
$0n_0+1n_1+\dots+pn_p=M$)
\begin{align}
& \frac{(2l{+}5)!!}{(l+2)!}
\sum_{K\geq 0}
(N{-}2{+}K)! \frac{B_{N{-}M{-}l{-}4,K}(\{x_r\}) }{(N{-}M{-}l{-}4)!}
\label{Bell-l}
\\[-1ex]
&-\sum_{K\geq 0} (N{-}3{+}K)!\frac{B_{N{-}M{-}l{-}4,K}(\{x_r\}) }{(N{-}M{-}l{-}4)!}
\sum_{i=0}^{p} n_i \frac{(2l+2i+3)!!(2i+1) i!}{
(2i+1)!!(l+i+1)!}
\nonumber
\\
&=\sum_{j\geq 1} \sum_{K\geq 0} (N{-}2{+}K)!
\frac{(2j{+}2l{+}5)!!(j{+}1)!}{(2j{+}1)!!(j{+}l{+}2)!}
\cdot \frac{x_j}{j!} \cdot \frac{B_{N{-}M{-}l{-}j{-}4,K}(\{x_r\}) }{
(N{-}M{-}l{-}j{-}4)!}
\nonumber
\\
&+\frac{1}{2}\sum_{l'=0}^{l} \sum_{n_0'=0}^{n_0} \dots 
\sum_{n_p'=0}^{n_p} 
\frac{(2l'{+}1)!!(2l''{+}1)!! }{l'!\, l''!} 
\binom{n_0}{n_0'}\cdots \binom{n_p}{n_p'}
\nonumber
\\
& 
\times \!\!\! \sum_{K',K''\geq 0} \!\!
(N'{-}2{+}K')!
\frac{B_{N'{-}M'{-}l'{-}2,K'}(\{x_r\})}{
(N'{-}M'{-}l'{-}2)!}
(N''{-}2{+}K'')!
\frac{
B_{N''{-}M''{-}l''{-}2,K''}(\{x_r\})}{(N''{-}M''{-}l''{-}2)!},
\nonumber
\end{align}
where $N':=
n_0'+\dots+n_p'$ and 
$M':=0n_0'+1n_1'+\dots+pn_p'$ as well as $l'':=l-l'$,
$N'':=N-N'$ and $M'':=M-M'$.
The sums over $j,K,K',K''$ are restricted to the range of non-trivial 
Bell polynomials and inverse Gamma functions.
\end{conj}

\noindent
We have checked (\ref{Bell-l}) with a computer algebra program 
for many different $l,p,n_i$.
Of course a direct
proof will be necessary\footnote{Other identities found during this
  work include 
for any $m,p,n_2,\dots,n_p \in \mathbb{N}$:
\begin{align*}
&
\sum_{n_i'+n_i''=n_i} \;\sum_{k'+k''=m}
\frac{(2k'{+}1)!!(2k''{+}1)!!
(k'+\sum_{j=2}^p jn'_j)!
(k''+\sum_{j=2}^p jn''_j)!}{
k'!k''!
(2{+}k'+\sum_{j=2}^p (j {-}1)n'_j)!
(2{+}k''+\sum_{j=2}^p (j {-}1)n''_j)!}
\prod_{j=2}^p \binom{n_j}{n_j'}
\nonumber
\\
&
=\frac{2 \cdot (m{+}1{+}\sum_{j=2}^p jn_j)!} {(m{+}4{+}\sum_{j=2}^p (j {-}1)n_j)!}
\Big\{\frac{(2m{+}3)!!}{m!}
+
\sum_{j=2}^{p} n_j\Big(
\frac{(2m{+}3)!!}{(m{+}2)! } ((m{+}3)j {+}m{+}2)
-\frac{j!(2j{+}2m{+}3)!!}{(j{+}m{+}1)!(2j{-}1)!! }
\Big) \Big\}.
\end{align*}
}.

The generating function of Bell polynomials is
\begin{align}
\exp\Big(u\sum_{j=1}^\infty \frac{x_jt^j}{j!}\Big) =\sum_{n,k\geq 0}
u^k \frac{t^n}{n!} B_{n,k}\big(\{x_r\}_{r=1}^{n-k+1}\big).
\end{align}
Multiplying by $e^{-u}u^{B-3}$, integrating over $u\in \mathbb{R}_+$
and differentiating with respect to $t$ gives an alternative
realisation of (\ref{solution:gammaMB}), where we also insert
the definition (\ref{def-rhol}) of $\rho_r$. With the series
$\sum_{j=1}^\infty \frac{(2j+1)!!}{(j+1)!}y^j=
\frac{1}{y} \sum_{k=2}^\infty \binom{-\frac{1}{2}}{k} (-2y)^{k}
=\frac{1}{y} (\frac{1}{\sqrt{1-2y}}-1-y)
=
\frac{2}{(1 + \sqrt{1-2y})\sqrt{1-2y}}-1$, below with $y=\frac{t}{T+c}$,
we arrive at
\begin{align}
&\rho_0^{B-3}M!(B-3-M)!\gamma^M_B
\nonumber
\\
&= \int_0^\infty du \;e^{-u} u^{B-3}
\frac{d^{B-3-M}}{dt^{B-3-M}}
\exp\Big(\frac{u}{\rho_0} \sum_{r=1}^\infty \frac{t^r(2r+1)!!}{(r+1)!}
(-\rho_r)\Big)\Big|_{t=0}
\nonumber
\\
&= \int_0^\infty du \;e^{-u} u^{B-3}
\nonumber
\\
&\times \frac{d^{B-3-M}}{dt^{B-3-M}}
\exp\Big(\frac{u}{2\rho_0} \int_1^\Xi \frac{dT
  \rho(T)}{\sqrt{T{+}c}^3} \Big(\frac{2(T{+}c)}{(\sqrt{T{+}c}+\sqrt{T{+}c{-}2t})
\sqrt{T{+}c{-}2t}}-1\Big)
\Big)\Big|_{t=0}
\nonumber
\\
&=
\rho_0^{B-2}
\frac{d^{B-3-M}}{dt^{B-3-M}}
\frac{(B-3)!}{\big(\frac{1}{\sqrt{Z}}-
\int_1^\Xi \frac{dT
  \rho(T)}{\sqrt{T+c}} \frac{1}{(\sqrt{T+c}+\sqrt{T+c-2t})
\sqrt{T+c-2t}}
\big)^{B-2}}\Big|_{t=0}.
\end{align}

Combined with the ansatz (\ref{goodansatz}) and with $Z=1$ in 2
dimensions we have proved
(provided that Conjecture~\ref{conj} is true):
\begin{thm} 
\label{thm:G1B}
The $(1+\dots+1)$-point function 
with $B\geq 3$ boundary
  components of the $\Phi^3_2$ matricial QFT-model has the solution
\begin{align}
G(X^1|\dots|X^B)
&=
(-2\tilde{\lambda})^{3B-4}
\frac{d^{B-3}}{dt^{B-3}}
\Bigg(\frac{\frac{1}{\sqrt{X^1+c-2t}^3}
\cdots \frac{1}{\sqrt{X^B+c-2t}^3}
}{\big(1-
\int_1^\infty \frac{dT
  \rho(T)}{\sqrt{T+c}} \frac{1}{(\sqrt{T+c}+\sqrt{T+c-2t})
\sqrt{T+c-2t}}
\big)^{B-2}}
\Bigg)\Bigg|_{t=0}.
\label{G1B-thm}
\end{align}
\end{thm}

\noindent
Together with (\ref{GNB-final-X}) we have thus completely solved
the combined large-($\mathcal{N},V$) limit of the Kontsevich
model.

\section{From $\Phi^3_2$ model on Moyal space to 
Schwinger functions on $\mathbb{R}^2$}

\label{sec:NCG}

This section parallels the treatment of the $\phi^{\star 4}_4$ case in
\cite{Grosse:2012uv}. We refer to that paper for more details.
The $\phi^{\star 3}_2$-model on Moyal-deformed 2D Euclidean space with
harmonic propagation is defined by the action
\begin{align}
S[\phi]:= \int_{\mathbb{R}^2} \frac{d\xi}{8\pi}
\Big( \kappa \phi+ \frac{1}{2}
\phi \star (-\Delta + \|4 \Theta^{-1} \cdot \xi\|^2 +\mu^2) \phi
+\frac{\lambda}{3} \phi\star\phi\star\phi\Big)(\xi).
\label{action-Moyal}
\end{align}
The tadpole contribution proportional to $\kappa\in \mathbb{R}$ is required
for renormalisation.
By $\star$ we denote the 2D-Moyal product parametrised
by $\theta \in \mathbb{R}$,
\begin{align}
(f\star g)(\xi) := \int_{\mathbb{R}^2 \times\mathbb{R}^2} \frac{d\eta
  \,dk}{(2\pi)^2} \,f(\xi+\tfrac{1}{2}\Theta\cdot k)\,g(\xi+\eta)
e^{\mathrm{i}\langle k,\eta\rangle},\qquad
\Theta:=\begin{pmatrix}0 & \theta \\ -\theta & 0
\end{pmatrix}.
\end{align}
The Moyal space possesses a convenient matrix basis
\begin{align}
f_{mn}(\xi)=
2 (-1)^m \sqrt{\frac{m!}{n!}}\Big(\sqrt{\frac{2}{\theta}}
\xi\Big)^{n-m}
L^{n-m}_m\Big(\frac{2\|\xi\|^2)}{\theta}\Big) e^{-\frac{\|\xi\|^2}{\theta}},
\quad m,n\in \mathbb{N},\label{fmn}
\end{align}
where the $L^\alpha_m(t)$ are associated Laguerre polynomials
of degree $m$ in $t$ and
$(\xi_1,\xi_2)^k:=(\xi_1+\mathrm{i}\xi_2)^k$. The matrix basis
satisfies
$(f_{kl}\star f_{mn})(\xi)=\delta_{ml}
f_{kn}(\xi)$ and $\int_{\mathbb{R}^2}d\xi
\;f_{mn}(\xi)=(2\pi\theta) \delta_{mn}$. A convenient
regularisation consists in
restricting the fields $\phi$ to those with finite expansion
$\phi(\xi)=\sum_{m,n=0}^{\mathcal{N}} \Phi_{mn}
f_{mn}(\xi)$. Using formulae for Laguerre polynomials, the action
(\ref{action-Moyal}) takes precisely the form (\ref{action-MM}) 
of a matrix model for $\phi =\phi^*\in M_{\mathcal{N}}( {\mathbb C})$,
with the following identification:
\begin{align}
V=\frac{\theta}{4},\qquad E_m= \frac{m}{V}+\frac{\mu^2}{2}
=\mu^2\Big(\frac{1}{2}+\frac{m}{\mu^2V}\Big).
\end{align}
This explains our interest in linearly spaced eigenvalues $e(x)=x$.

Following \cite{Grosse:2013iva} we \emph{define} connected Schwinger
functions in position space as
\begin{align}
S_c(\mu \xi_1,{\dots},\mu \xi_N)
&:= \!\!\!\!\!\lim_{V\mu^2\to \infty}\lim_{\Lambda\to \infty}
 \sum_{N_1+\dots+N_B=N}
\sum_{q^1_1,\dots, q^B_{N_B} =0}^{\mathcal{N}}
\frac{G_{|q^1_1\dots q^1_{N_1}|\dots
|q^B_1\dots q^B_{N_B}|}}{8\pi \mu^{2(2-B-N)}
S_{(N_1,\dots,N_B)}}
\nonumber
\\*
& \times\!\!\!
\sum_{\sigma \in \mathcal{S}_N} \prod_{\beta=1}^B
\frac{
f_{q_1^\beta q_2^\beta}(
\xi_{\sigma(s_\beta{+}1)})
{\cdots}
f_{q^\beta_{N_\beta}q^\beta_1}(
\xi_{\sigma(s_\beta+N_\beta)})
}{V\mu^2 N_\beta} ,
\label{Schwinger}
\end{align}
where $s_\beta:=N_1{+}{\dots}{+}N_{\beta{-}1}$ and
$\mathcal{N}=\Lambda^2V\mu^2$.  The $G_{\dots}$
are the expansion coefficients of $\log\mathcal{Z}[J]$ in (\ref{logZ}),
where we already absorbed their mass dimension 
given in footnote~\ref{fn-1}. These
Schwinger functions are fully symmetric in $\mu \xi_1,\dots,\mu \xi_N$.

The various factors of $V$ need explanation.
We recall that the prefactor of $G_{\dots}$ in (\ref{logZ})
was $V^{2-B}$. The factor $V^{-B}$ is distributed over the $B$
cycles. In a first step we have thus defined the free energy
density as $(\mu^2V)^{-2} \log \frac{\mathcal{Z}[J]}{\mathcal{Z}[0]}$, in
agreement with the usual procedure in matrix models 
(see e.g.\ the $\frac{1}{N^2}$ prefactor 
in \cite[eq. (4.2)]{Makeenko:1991ec}). Then formally we set
\[
S_c(\mu \xi_1,{\dots},\mu \xi_N)=\frac{1}{8\pi}
\frac{\delta^N ((\mu^2V)^{-2}
\log \frac{\mathcal{Z}[J]}{\mathcal{Z}[0]})}{
\delta J(\xi_1)\dots \delta J(\xi_N)} \Big|_{J=0},
\]
with a \emph{special} definition of $\frac{\delta J_{mn}}{\delta
  J(\xi)}$. Since by properties of the matrix basis (\ref{fmn}) one has
$J_{mn}=\int_{\mathbb{R}^2} \frac{d\eta}{8\pi V} f_{nm}(\eta)J(\eta)$,
the usual convention
$\frac{\delta J(\eta)}{\delta J(\xi)}=\delta(\xi-\eta)$ gives
$\frac{\delta J_{mn}}{\delta  J(\xi)}= \frac{1}{8\pi V} f_{nm}(\xi)$. As
part of the renormalisation process, we change these conventions into
\begin{align}
\frac{\delta J_{mn}}{\delta  J(\xi)}:= \mu^2 f_{nm}(\xi),
\end{align}
or equivalently $S_c(\mu \xi_1,{\dots},\mu \xi_N)=\frac{1}{8\pi}
\frac{(8\pi V\mu^2)^N \delta^N ((\mu^2V)^{-2}
\log \frac{\mathcal{Z}[J]}{\mathcal{Z}[0]})}{
\delta J(\xi_1)\dots \delta J(\xi_N)} \Big|_{J=0}$ with the standard convention.
It is important to note that these field redefinitions are neutral
with respect to the number $B$ of boundary components.

The evaluation of (\ref{Schwinger}) follows the same lines as in
\cite{Grosse:2013iva}.  To keep this paper self-contained, we outline
the steps until the technical lemma proved in \cite[Lemma 4+Corollary
5]{Grosse:2013iva} can be used.  We collect the indices
$\underline{q}^\beta:=(q^\beta_1,\dots,q^\beta_{N_\beta})$ and define
$|\underline{q}^\beta|:=q^\beta_1+\dots+q^\beta_{N_\beta}$ and $\langle
\underline{\omega}^\beta,\underline{q}^\beta\rangle :=
\sum_{i=1}^{N_\beta-1} \omega^\beta_i (q^\beta_i-q^\beta_{i+1})$ for
$\underline{\omega}=(\omega^\beta_1,\dots,\omega^\beta_{N_\beta-1})$.
We assume that the matrix functions $G$ have a representation as
Laplace-Fourier transform,
\begin{align}
\frac{G_{|\underline{q}^1|\dots|\underline{q}^B|}}{\mu^{2(2-B-N)}}
= \int_{\mathbb{R}_+^B} d(t^1,\dots,t^B)
\int_{\mathbb{R}^{N-B}} &d(\underline{\omega}^1,
\dots ,\underline{\omega}^B)
\mathcal{G}_{\mathcal{N},V}
(t^1,\underline{\omega}^1|\dots
|t^B,\underline{\omega}^B)
\label{Laplace-Fourier}
\\*[-2ex]
&\times
\exp\Big(-\frac{1}{V\mu^2} \sum_{\beta=1}^B
\big(t^\beta|\underline{q}^\beta|-\mathrm{i}\langle\underline{\omega}^\beta,\underline{q}^\beta\rangle \big)\Big).
\nonumber
\end{align}
The inverse Laplace-Fourier transforms 
$\mathcal{G}_{\mathcal{N},V}(t^1,\underline{\omega}^1|\dots
|t^B,\underline{\omega}^B)$ depend on $\mathcal{N},V$ but have a limit
$\mathcal{G}(t^1,\underline{\omega}^1|\dots
|t^B,\underline{\omega}^B)= \lim_{ \mathcal{N},V\to \infty} 
\mathcal{G}_{\mathcal{N},V}(t^1,\underline{\omega}^1|\dots
|t^B,\underline{\omega}^B)$ 
satisfying
\begin{align}
\tilde{G}(\underline{x}^1|\dots|\underline{x}^B|)
= \int_{\mathbb{R}_+^B} d(t^1,\dots,t^B)
\int_{\mathbb{R}^{N-B}} &d(\underline{\omega}^1,
\dots ,\underline{\omega}^B)
\mathcal{G}(t^1,\underline{\omega}^1|\dots
|t^B,\underline{\omega}^B)
\label{Laplace-Fourier-lim}
\\*[-2ex]
&\times
\exp\Big(-\sum_{\beta=1}^B
\big(t^\beta|\underline{x}^\beta|-\mathrm{i}\langle\underline{\omega}^\beta,\underline{x}^\beta\rangle \big)\Big).
\nonumber
\end{align}
Inserting (\ref{Laplace-Fourier}) into (\ref{Schwinger}) gives, besides
$\mathcal{G}_ {\mathcal{N},V}(t^1,\underline{\omega}^1|\dots
|t^B,\underline{\omega}^B)$,
the following type of factors (for each $\beta=1,\dots,N_\beta$
omitted below) under the Laplace-Fourier integral and
the sum over permutations and partitions of $N$:
\begin{align}
&\sum_{q_1,\dots,q_N=0}^{\mathcal{N}}
\frac{f_{q_1q_2}(\xi_{\sigma(s+1)})\cdots
  f_{q_Nq_1}(\xi_{\sigma(s+N)})}{V\mu^2 N}
z_1^{q_1}(t,\underline{\omega})
\cdots z_N^{q_N}(t,\underline{\omega}),
\label{sum-Laguerre}
\\
&z_1=e^{-\frac{t}{V\mu^2}+\mathrm{i}\frac{\omega_1}{V\mu^2}},\quad
z_i=e^{-\frac{t}{V\mu^2}+\mathrm{i}\frac{\omega_i-\omega_{i-1}}{V\mu^2}}~~
\text{for } i=2,\dots,N-1,\quad
z_N=e^{-\frac{t}{V\mu^2}-\mathrm{i}\frac{\omega_{N-1}}{V\mu^2}}.
\nonumber
\end{align}
For $\mathcal{N}\to \infty$ but fixed $V$, the index sum was evaluated in
\cite{Grosse:2013iva}:
\begin{lem}[{\cite[Lemma~4+Corollary~5]{Grosse:2013iva}}]
Let $\langle \xi,\eta\rangle$, $\|\xi\|$ and
  $\xi{\times}\eta=\det(\xi,\eta)$ be scalar product, norm and (third component
  of) vector product of $\xi,\eta\in \mathbb{R}^2$. Then for $\xi_i \in
  \mathbb{R}^2$ and $z_i\in \mathbb{C}$ with $|z_i|<1$ one has
(with cyclic identification $N+i\equiv i$ where necessary)
\begin{align}
&\sum_{q_1,\dots,q_N=0}^\infty \frac{1}{V\mu^2} \prod_{i=1}^N
f_{q_iq_{i+1}}(\xi_i) z_i^{q_i}
\label{sumlim-Laguerre}
\\[-3ex]
&=\frac{2^N}{V\mu^2(1-\prod\limits_{i=1}^N (-z_i))}
\exp\bigg(-\frac{\sum\limits_{i=1}^N \|\xi_i\|^2}{4V}
\frac{1+\prod\limits_{i=1}^N (-z_i)}{1-\prod\limits_{i=1}^N (-z_i)}\bigg)
\nonumber
\\[-1.5ex]
& \times \exp\bigg(\!\!\!-\!
\!\!\! \sum_{1\leq k < l \leq N} \!\!\!
\Big(\!
\frac{\big(\langle \xi_k, \xi_l\rangle
{-}\mathrm{i} \xi_k{\times} \xi_l\big)}{2V}
\frac{\prod\limits_{j=k+1}^l ({-}z_j)}{1{-}\prod\limits_{i=1}^N ({-}z_i)}
+\frac{\big(\langle \xi_k, \xi_l\rangle
{+}\mathrm{i} \xi_k{\times} \xi_l\big)}{2V}
\frac{\prod\limits_{j=l+1}^{N+k} ({-}z_j)}{1{-}\prod\limits_{i=1}^N ({-}z_i)}
\! \Big)\!\bigg).
\nonumber
\end{align}
\end{lem}
That the result can be
applied to the combined limit $\mathcal{N},V\to \infty$ with
$\mathcal{N}=\Lambda^2V\mu^2$, where $|z_i|=1$ becomes critical, 
needs some explanation. It is uncritical to move the convergent
$\mathcal{G}(t^1,\underline{\omega}^1|\dots
|t^B,\underline{\omega}^B)$ in front of the limit. The result
(\ref{sumlim-Laguerre}) relies on the generating function
$\sum_{n=0}^\infty L^{\alpha-n}_n(t)z^n=e^{-zt}(1+z)^\alpha$ which
\emph{precisely for $\alpha\in \mathbb{N}$}
is absolutely convergent for \emph{any} $z\in \mathbb{C}$. The only place
where $|z|<1$ matters is a final sum
$\sum_{q=0}^\infty \frac{(q+k)!}{q!k!}((-z_1)\cdots (-z_N))^{q}
=\frac{1}{(1-(-z_1)\cdots (-z_N))^{1+k}}$. Restricting this sum to
$1\leq \mathcal{N}$ gives (for $N$ being even) instead
\[
\sum_{q=0}^\mathcal{N}\frac{(q+k)!}{q!k!}( z_1\cdots z_N)^q
=\frac{1-(z_1\cdots z_N)^{\mathcal{N}+1}P_k(z_1\cdots z_N)}{
(1-z_1\cdots z_N)^{1+k}},
\]
where $P_k(z)$ is a polynomial of degree $k$ with $P_k(1)=1$. 
Since $(z_1\cdots z_N)^{\mathcal{N}}= e^{-\Lambda^2 N t}$, there is a
$V$-uniform multiplicative error of $1+\mathcal{O}(1)e^{-\Lambda^2 N t}$ if we
restrict in (\ref{sumlim-Laguerre}) the sum to $q_i\leq \mathcal{N}$. 
Therefore,
the limit $\lim_{V\to \infty}$ of (\ref{sumlim-Laguerre})
agrees with the scaling limit
 $\mathcal{N},V\to \infty$ with
$\frac{\mathcal{N}}{V\mu^2}=\Lambda^2$ fixed
of (\ref{sum-Laguerre}) followed by sending
$\Lambda \to \infty$. We thus have
\begin{align}
&\lim_{\Lambda\to \infty} \Bigg(\lim_{\stackrel{\mathcal{N},V\to \infty}{
\frac{\mathcal{N}}{V\mu^2}=\Lambda^2}}\;
\sum_{q_1,\dots,q_N=0}^{\mathcal{N}}
\frac{f_{q_1q_2}(\xi_{\sigma(s+1)})\cdots
  f_{q_Nq_1}(\xi_{\sigma(s+N)})}{V\mu^2 N}
z_1^{q_1}(t,\underline{\omega})
\cdots z_N^{q_N}(t,\underline{\omega})
\Bigg)
\nonumber
\\
&= \left\{
\begin{array}{c@{\qquad\text{for }}l} 0 & N \text{ odd,} \\
 \frac{2^N}{N^2 t}
\exp\Big({-}\frac{\mu^2}{2Nt}\|\xi_{\sigma(s+1)}{-}
\xi_{\sigma(s+2)}{+}\dots {-} \xi_{\sigma(s+N)}\|^2\Big)
& N \text{ even.}
\end{array}\right.
\label{sum-Laguerre-final}
\end{align}
Now write
\begin{align}
\frac{2^N}{N^2 t}
e^{-\frac{\mu^2}{2Nt}\|\xi\|^2}
=\frac{2^N}{2\pi N}
\int_{\mathbb{R}^2} dp \; e^{-\frac{N}{2\mu^2}\|p\|^2 t + \mathrm{i}
\langle p,\xi\rangle}
\end{align}
and recall that the $z_i$ factors of (\ref{sum-Laguerre-final}) were
introduced via the Laplace-Fourier transform
(\ref{Laplace-Fourier}) to be inserted into (\ref{Schwinger}).
Combining all these steps and limits,
we can immediately perform the Laplace-Fourier transform
(\ref{Laplace-Fourier-lim}) to a function with arguments
$x^\beta_i=\frac{\|p_i\|^2}{2\mu^2}$ for all $i=1,\dots,N_\beta$.
The final result reads
\begin{align}
S_c(\mu \xi_1,{\dots},\mu \xi_N)
&= \!\!\!\!\!
\sum_{\stackrel{N_1+\dots+N_B=N}{N_\beta \text{ even}}}
\sum_{\sigma \in\mathcal{S}_N}
\prod_{\beta=1}^B \Big(\frac{2^{N_B}}{N_B}
\int_{\mathbb{R}^2} \frac{dp^\beta}{2\pi \mu^2} e^{\mathrm{i}\langle
    p^\beta,\xi_{\sigma(s_\beta+1)}-
\xi_{\sigma(s_\beta+2)}+\dots - \xi_{\sigma(s_\beta+N_\beta)}\rangle}\Big)
\nonumber
\\
& \times \frac{1}{8\pi S_{(N_1,\dots,N_B)}}\;
\tilde{G}\Big(\underbrace{\tfrac{\|p^1\|^2}{2\mu^2},\dots,
\tfrac{\|p^1\|^2}{2\mu^2}}_{N_1}\big|\dots\big|
\underbrace{\tfrac{\|p^B\|^2}{2\mu^2},\dots,
\tfrac{\|p^B\|^2}{2\mu^2}}_{N_B}\Big).
\label{Schwinger-final}
\end{align}

For $N=2$ the formula specifies with (\ref{Gxy}) and
(\ref{MS-solution}) to
\begin{align}
&S_c(\mu \xi_1,\mu \xi_2)
= \int_{\mathbb{R}^2} \frac{dp^\beta}{4\pi^2 \mu^2}
e^{\mathrm{i}\langle p, \xi_1-\xi_2\rangle}\hat{S}_2(p),
\label{Schwinger-2}
\\
&\hat{S}_2(p) = 2
W'\Big(\big(\tfrac{\|p\|^2}{\mu^2}+1\big)^2\Big)
=
\frac{
\displaystyle
1- \tilde{\lambda}^2 \int_1^\infty \!\! \frac{dT}{\sqrt{T}\sqrt{T{+}c}}
\frac{\mu^4}{
\big(\sqrt{(\|p\|^2{+}\mu^2)^2{+}c\mu^4}+\mu^2 \sqrt{T{+}c}\big)^2}
}{\sqrt{(\|p\|^2{+}\mu^2)^2+c\mu^4}}.
\nonumber
\end{align}

It was also pointed out in \cite{Grosse:2013iva} and
\cite{Grosse:2014lxa} that the Schwinger 2-point function is
reflection positive iff the function $\|p\|^2 \mapsto \hat{S}_2(p)$ is
a Stieltjes function. This is not the case, neither for real nor
purely imaginary non-vanishing $\tilde{\lambda}$! For $c>0$ and thus
$\tilde{\lambda} \in \mathrm{i}\mathbb{R}$, the integrand has a pole
(or end point of a branch cut)
in the complex plane at $\|p\|^2=\mu^2(-1\pm \mathrm{i} \sqrt{c})$,
contradicting holomorphicity in $\mathbb{C}\setminus
\mathbb{R}_-$. For $-1<c<0$ and thus 
$\tilde{\lambda} \in \mathbb{R}$ one finds that 
the imaginary part of
$\hat{S}_2(p)$ at $\|p\|^2=(-3-\mathrm{i}\frac{|c|}{10})\mu^2$ is
negative\footnote{Here one should write $\sqrt{(\|p\|^2+\mu^2)^2+c\mu^4}
\mapsto \sqrt{\|p\|^2+(1-\sqrt{-c})\mu^2}
\sqrt{\|p\|^2+(1+\sqrt{-c})\mu^2}$ for a well-defined 
holomorphic extension of (\ref{Schwinger-2}).}.  
This
contradicts the anti-Herglotz property of Stieltjes functions.
A rigorous proof that the 2-point function of $\Phi^3_2$ is not
reflection positive will be given in 
\cite{Grosse:2016qmk}.

\section{Summary}

\noindent 
We have given an alternative solution strategy for the
large-$\mathcal{N}$ limit of the $\Phi^3_2$
matrix model (= renormalsed Kontsevich model). This limit suppresses
non-planar graphs. In principle, punctures (or boundary components)
are also suppressed, but special limits of noncommutative field theory
amplify them to the same level as the disk topology. We have
established exact formulae, analytic in the (squared) coupling
constant, for all these correlation functions. 
Correlation functions of disk topology (single puncture) can certainly
be derived from previous results on the Kontsevich model. The
complete treatment of the multi-punctured cases is new 
(to the best of our knowledge).

In our subsequent paper \cite{Grosse:2016qmk} we extend this work to
the $\Phi^3_4$ and $\Phi^3_6$ models. There the renormalisation is
much more involved, whereas the solution of Schwinger-Dyson equations
is easily adapted from $\Phi^3_2$. We will discuss the issue of
overlapping divergences and renormalons in $\Phi^3_6$. The main result
will be the proof that $\Phi^3_4$ and $\Phi^3_6$, but not $\Phi^3_2$,
have reflection positive 2-point functions.

Reflection posivity of higher correlation functions is work in progress.
Another interesting question concerns the identification of the KdV
hierarchy in the solution we found. 

We also hope that these investigations provide new ideas for attacking
the more difficult equations of the $\Phi^4_4$ model.

\section*{Acknowledgements} 

\noindent This work started in collaboration with Ricardo Kullock
from Brazil (now at Universidade do Estado do Rio de Janeiro) during his stay
in Vienna. We would like to thank him for his contributions and an
enjoyable collaboration.
\noindent A.S.\ was supported by JSPS
KAKENHI Grant Number 16K05138, and R.W.\ by SFB 878.


\end{document}